%% file: main_icml.tex
\documentclass{article}
\usepackage{styling_pods}
\usepackage{natbib}
\usepackage[margin=1in]{geometry}

\newcommand{\OL}{\mathsf{OL}}

\newcommand{\linf}[1]{\|#1\|_{\infty}}
\newcommand{\lp}[1]{\|#1\|_{p}}
\renewcommand{\epsilon}{\varepsilon}

\newcommand{\Proj}{\mathbb{P}}

\usepackage[ruled,linesnumbered, vlined]{algorithm2e}
\title{High-Dimensional Geometric Streaming for Nearly Low Rank Data}
\date{}

\author{Hossein Esfandiari\footnote{Google Research. \texttt{esfandiari@google.com}.} \and  Praneeth Kacham\footnote{Google Research. \texttt{pkacham@google.com}. Work done while the author was a student at Carnegie Mellon University.} \and Vahab Mirrokni\footnote{Google Research. \texttt{mirrokni@google.com}.} \and David P. Woodruff\footnote{Carnegie Mellon University. \texttt{dwoodruf@cs.cmu.edu}.} \and Peilin Zhong\footnote{Google Research. \texttt{peilinz@google.com}.}}
\begin{document}
\maketitle
\begin{abstract}
We study streaming algorithms for the $\ell_p$ subspace approximation problem. Given points $a_1, \ldots, a_n$ as an insertion-only stream and a rank parameter $k$, the $\ell_p$ subspace approximation problem is to find a $k$-dimensional subspace $V$ such that $(\sum_{i=1}^n d(a_i, V)^p)^{1/p}$ is minimized, where $d(a, V)$ denotes the Euclidean distance between $a$ and $V$ defined as $\min_{v \in V}\opnorm{a - v}$. When $p = \infty$, we need to find a subspace $V$ that minimizes $\max_i d(a_i, V)$.  For $\ell_{\infty}$ subspace approximation, we give a deterministic strong coreset construction algorithm and show that it can be used to compute a $\poly(k, \log n)$ approximate solution. We show that the distortion obtained by our coreset is nearly tight for any sublinear space algorithm. For $\ell_p$ subspace approximation, we show that suitably scaling the points and then using our $\ell_{\infty}$ coreset construction, we can compute a $\poly(k, \log n)$ approximation. Our algorithms are easy to implement and run very fast on large datasets. We also use our strong coreset construction to improve the results in a recent work of Woodruff and Yasuda (FOCS 2022) which gives streaming algorithms for high-dimensional geometric problems such as width estimation, convex hull estimation, and volume estimation. 

\end{abstract}

\section{Introduction}
Modern datasets are usually very high-dimensional and have a large number of data points. Storing the entire dataset to analyze them is often impractical and in certain settings impossible. In recent years, streaming algorithms have emerged as a way to process and understand the datasets in both a space and time-efficient manner. In a single-pass streaming setting, an algorithm is allowed to make only a single pass over the entire dataset and is required to output a ``summary'' of the dataset that is useful to solve a certain problem. In this work, we focus on streaming algorithms for high-dimensional geometric problems such as subspace approximation, width estimation, etc. Suppose we are given a set of $d$-dimensional points $a_1,\ldots,a_n$ and an integer parameter $k \le d$. Given a subspace $V$, we define $d(a, V)$ to be distance between the point $a$ and subspace $V$ given by $\min_{v \in V}\opnorm{a-v}$. The $\ell_p$ subspace approximation problem \citep{deshpande2011subspace}, for $p \in [1, \infty]$, asks to find a $k$-dimensional subspace that minimizes $(\sum_{i=1}^n d(a_i, V)^p)^{1/p}.$

Note that for $p = \infty$, we want to find a $k$-dimensional subspace that minimizes the maximum distance from the given set of points. Related to the $\ell_{\infty}$ subspace approximation problem is the widely studied outer $(d-k)$ radius estimation problem \citep{varadarajan2007radii} which instead asks for a $k$-dimensional flat\footnote{A $k$-dimensional flat is defined as a $k$ dimensional subspace that is translated by some $c$.} $F$ that minimizes $\max_{i \in [n]}d(a_i, F)$. The outer $(d-k)$ radius is a measure of how far the point set is from being inside a $k$-dimensional flat. \citet{varadarajan2007radii} give a polynomial time algorithm for approximating the outer $(d-k)$ radius up to an $O(\sqrt{\log n})$ multiplicative factor. Their algorithm is based on rounding of  a semidefinite program (SDP) relaxation. When $n$ is very large, their algorithm is not practical and cannot be implemented in the streaming setting. We give a time and space-efficient single pass streaming algorithm that approximates the outer $(d-k)$ radius up to an $\tilde{O}(\sqrt{k}\log(n\kappa))$ factor, where $\kappa$ is a suitably defined condition number. Typically, the value of $k$ used is much smaller than $n$ and $d$ since in many settings, we have that the $n \times d$ matrix $A$ is a noisy version of an underlying rank $k$ matrix, for a small value of $k$.

Our main contribution is a simple \textbf{\emph{deterministic}} algorithm that constructs a \emph{strong coreset} for approximating $\max_i d(a_i, V)$ for any $k$-dimensional subspace $V$ in a single-pass streaming setting. We note that this notion of strong coreset is different from the strong/weak coreset definitions in some computational geometry works. When run on the stream of points $a_1,\ldots, a_n$, our algorithm selects a subset $S \subseteq [n]$ of points with $|S| = O(k\log^2(n\kappa))$, such that for all $k$-dimensional subspaces $V$,
$
    \max_{i \in S} d(a_i, V) \le \max_{i \in [n]}d(a_i, V) \le O(\sqrt{k}\log(n\kappa))\max_{i \in S}d(a_i, V).
$
We stress that our coreset can be used to approximate the max distance of the point set to \emph{any} $k$-dimensional subspace and hence it is termed a strong coreset.
We prove:
\begin{theorem}[Informal]
   Given a parameter $k$ and $n$  points $a_1,\ldots,a_n \in \R^d$, Algorithm~\ref{alg:efficient} selects a subset $S \subseteq [n]$ of points with $|S| = O(k\log^2 n\kappa)$, such that for all $k$-dimensional subspaces $V$,
   \begin{align*}
       1 \le \frac{\max_{i \in [n]}d(a_i, V)}{\max_{i \in S}d(a_i, V)} \le O(\sqrt{k}\log n\kappa).
   \end{align*}
   The streaming algorithm requires only enough space to store $O(k\log^{2} n\kappa)$ rows of $A$ and can be implemented in time $O(\nnz(A)\log n + d\poly(k, \log n\kappa))$ if one is allowed randomization. 
   \label{thm:intro-theorem}
\end{theorem}
In this result and its applications throughout the paper, the condition number $\kappa$ can be replaced with $n^{O(k)}$ assuming that all the entries in the input points are integers bounded in absolute value by $\poly(n)$. We note that some assumption on bit complexity is necessary in order to establish memory bounds in a streaming setting. Under suitable assumptions about the ``noise'' in the process generating the data, $\kappa$ can be much smaller than $n^{O(k)}$. 

We then show using a simple reduction that the above theorem can be used to approximate the outer $(d-k)$ radius by running the streaming algorithm on the point set $a_2 - a_1, \ldots, a_n - a_1$.

We also prove the following lower bound showing that our coreset obtains near-optimal distortion up to logarithmic factors in $n$ and $\kappa$.
\begin{theorem}[Informal]
Given parameters $n$, $d$ and $k$ with $k = \Omega(\log n)$, any streaming algorithm that computes a strong coreset with distortion at most $O(\sqrt{k/\log n})$ with probability $\ge 9/10$ must use $\Omega(n)$ bits of space.
\end{theorem}

We then turn to the $\ell_p$ subspace approximation problem for general $p \in [1,\infty)$. We observe that an instance of the $\ell_p$ subspace approximation problem can be turned into an $\ell_{\infty}$ subspace approximation problem by using the so-called min-stability property of exponential random variables. We scale each input point with appropriately chosen independent random variables and feed the scaled points to Algorithm~\ref{alg:efficient}. We obtain the following result: 
\begin{theorem}[Informal]
    Given $p \ge 1$, a dimension parameter $k$, and $n$ points $a_1,\ldots,a_n \in \R^d$, there is a randomized streaming algorithm that selects a subset $S \subseteq [n]$, $|S| = O(k \log^2 n\kappa)$ and assigns a weight $w_i \ge 0$ for $i \in S$ such that if
    \[
        \tilde{V} = \argmin_{k\text{-dim V}}\max_{i \in S}w_i \cdot d(a_i, V),
    \]
    then 
    {\[\frac{(\sum_{i=1}^n d(a_i, \tilde{V})^p)^{1/p}}{\min_{k\text{-dim V}}(\sum_{i=1}^n d(a_i, V)^p)^{1/p}} \le k^{1/2+2/p}\poly(\log^{1+3/p} n\kappa).\]}
    The algorithm only uses $O(d \cdot k \log^2 n\kappa)$ bits of space and runs in $O(\nnz(A)\log n + d\poly(k, \log n))$ time.
\end{theorem}
While exponential random variables have been previously used in the context of $\ell_p$ subspace embeddings and $\ell_p$ moment estimation in streams, as far as we are aware, ours is the first work to use them in the context of subspace approximation.

We then show that recent algorithms of \cite{wy22} can be improved using our coreset construction algorithm when the data points $a_1,\ldots,a_n$ are ``approximately'' spanned by a low rank subspace. They give streaming algorithms for a host of geometric problems such as width estimation, volume estimation, L\"owner-John ellipsoid computation, etc. The main ingredient of their algorithms is a deterministic $\ell_{\infty}$ subspace embedding: their algorithm streams through the rows of an $n \times d$ matrix $A$ and selects a subset $S \subseteq [n]$ of rows, $|S| = O(d\log n)$ with the property that for all $x$,
\begin{align*}
    \linf{A_S x} \le \linf{Ax} \le \sqrt{d\log n}\linf{A_S x}.
\end{align*}
Here $\linf{x} \coloneqq \max_i |x_i|$ and $A_S$ is the matrix $A$ restricted to only those rows in $S$. When the matrix $A$ has rank $d$, their algorithm necessarily needs $\Omega(d^2)$ bits of space which is prohibitive when $d$ is very large. In practice, many matrices $A$ are very well-approximated by a matrix with far lower rank than $d$ even when the rank of the matrix $A$ is $d$. Suppose $A$ is well-approximated by a rank $k$ matrix in the sense that there is a $k$-dimensional subspace $V$ such that all the rows of $A$ are not very far from $V$. We show that if $S$ is the coreset constructed by Algorithm~\ref{alg:efficient}, then for all \emph{unit vectors} $x$,
$\linf{A_Sx} \le \linf{Ax} \le (C\sqrt{k}\log n\kappa)\linf{A_S x} + C\Delta\log n\kappa$, where $\Delta$ denotes the optimal rank-$k$ $\ell_{\infty}$ subspace approximation cost of the matrix $A$. Thus, $\linf{A_Sx}$ can be used to approximate $\linf{Ax}$ well when $\Delta$ is small.

\subsection{Previous Work}
The rank-$k$ $\ell_{\infty}$ subspace approximation problem and more generally the rank-$k$ $\ell_{\infty}$ flat approximation problem have been previously studied for different values of $k$. As discussed earlier, \citet{varadarajan2007radii} give an SDP-based algorithm that can compute an $O(\sqrt{\log n})$ factor approximation for all values of $k$. Being SDP-based, the algorithm is impractical in the streaming setting and when the number of points $n$ is very large. We shall mostly discuss previous works relevant in the streaming setting.

For specific values of $k=0$ and $k=d-1$, \citet{agarwal2015streaming} study upper and lower bounds on streaming algorithms. For $k=0$, also known as the minimum enclosing ball (MEB) problem, they give a streaming algorithm that is a $(1 + \sqrt{3})/2$ approximation and show that there is a small enough constant $\alpha$ such that any $\alpha$ approximation algorithm must use $\min(n, \exp(d^{1/3}))$ space, thereby showing that there are no small-space streaming algorithms with a better than $\alpha$ approximation. For $k=d-1$, the so-called width estimation problem, they showed that any algorithm that approximates the cost up to a multiplicative $\Theta(d^{1/3})$ factor must use $\Omega(n, \exp(d^{1/3}))$ bits of space, again ruling out small-space algorithms with better than $d^{1/3}$ approximation factor.

Later, \citet{chan2014streaming} improved the approximation ratio of the algorithm of \citet{agarwal2015streaming} to $(1 + \sqrt{2})/2$ for the MEB problem.

Recently, \citet{tukan2022new} give an algorithm to construct a coreset for the $\ell_{\infty}$ subspace approximation problem with a size of $\tilde{O}(k^{3k})$. While an offline coreset construction can be converted into a streaming coreset construction using the merge-and-reduce procedure, the exponential dependence in $k$ makes their algorithm impractical compared to our algorithm which needs to store only $O(k \log(n\kappa)^2)$ input points.

For the $\ell_1$ subspace approximation problem, \citet{feldman2010coresets} give a streaming algorithm to construct a  coreset with $\tilde{O}\left(d\left(\frac{k \cdot 2^{O(\sqrt{\log n})}}{\varepsilon^2}\right)^{\poly(k)}\right)$ points that can be used to compute a $1+\varepsilon$ approximation. When $n$ and $d$ are large, the space requirement of the coreset is infeasible. In comparison, although our algorithms do not give $1+\varepsilon$ approximation, we can compute $\poly(k,\log n\kappa)$ approximations using only space necessary to store $\poly(k, \log n\kappa)$ points, which is much smaller than the coreset constructed by their algorithm.

For all values of $p$, \citet{kerber2014approximation} give a dimensionality reduction procedure by showing that projecting the points to a random $O(k^2(\log k/\varepsilon \cdot \log n)/\varepsilon^3)$-dimensional space preserves the $\ell_p$ subspace approximation\footnote{They prove their result for the more general problem of subspace clustering.} cost. For $p=\infty$, their algorithm combined with the coreset construction algorithm of \citet{wy22} can be used to approximate the $\ell_{\infty}$ subspace approximation up to $\poly(k, \log n)$ factors. But since the $d$-dimensional ``information'' is destroyed by the projection, we cannot recover a solution in the $d$-dimensional space. In comparison, for $p=\infty$, we give a practical algorithm to construct a strong coreset that lets us approximate the maximum distance to any $k$ dimensional subspace and for general $p$, we give a polynomial time algorithm that can output a ``$d$-dimensional'' approximate solution.

For $p \notin \set{1, 2, \infty}$, much less is known in the streaming setting. In the offline setting, \citet{deshpande2007sampling} gave a sampling based algorithm for all $p \ge 1$ that outputs a bicriteria solution for the $\ell_p$ subspace approximation problem. Later \citet{deshpande2011algorithms} gave a polynomial time $O(\sqrt{p})$ factor approximation algorithm for the $\ell_p$ subspace approximation problem for all $p \ge 2$. Assuming the Unique Games Conjecture, they show that it is hard to approximate the cost to a smaller than $O(\sqrt{p})$ factor. For $1 \le p \le 2$, \citet{robust-subspace-approximation} gave an input sparsity time algorithm that computes a $1+\varepsilon$ approximation but they have an $\exp(\poly(k/\varepsilon))$ term in their running time. The $O(\sqrt{p})$ factor approximation algorithm of \cite{deshpande2011algorithms} is based on convex relaxations and is not applicable in the streaming setting of this paper. In a recent work, \citet{deshpande2023one} observed the lack of streaming algorithms for $\ell_p$ subspace approximation that also have the subset selection property that our coresets have. They give a subset selection algorithm for the $\ell_p$ subspace approximation problem but their results have a weaker additive error guarantee. They leave open the subset selection algorithms that give a multiplicative approximation to the $\ell_p$ subspace approximation problem. In a recent work, \citet{woodruff2023new} answered the question of \citet{deshpande2023one} in the affirmative by giving a subset selection algorithm the computes a strong coreset with $O((k/\varepsilon)^{O(p)}\polylog(n))$ rows that can approximate the cost of any $k$-dimensional space up to a $1 \pm \varepsilon$ factor. Selecting $k^{O(p)}$ rows could be problematic when $p$ is large. Our work makes progress on this question by removing the exponential dependence in $p$,  although at the cost of only being able to compute a $\poly(k, \log n\kappa)$ approximation to the problem.

\paragraph{Relevance to Machine Learning.} Our work continues the long line of work in the area of subspace approximation and low rank approximation with different error metrics that has been of interest in the machine learning community. Previous works study problems such as $\ell_1$ subspace approximation \citep{hardt2013algorithms}, entrywise $\ell_p$ low rank approximation \citep{chierichetti2017algorithms, dan2019optimal}, column subset selection for the entrywise $\ell_p$ norm, and other error metrics \citep{song2019towards}. Our algorithms for geometric streaming problems such as convex hull estimation have applications to robust classification \citep{provost2001robust, fawcett2007pav}.

\section{Preliminaries}\label{sec:prelims}
For integer $n \ge 1$, we use $[n]$ to denote the set $\set{1,\ldots,n}$. For an $n \times d$ matrix $A$, we use $a_i \in \R^d$ to denote the $i$-th row. If $S \subseteq [n]$, then $A_S$ denotes the submatrix formed by the rows in the set $S$. Given indices $i < j$, we use $A_{i:j}$ to denote the matrix formed by the rows $a_i, \ldots, a_j$. For $x \in \R^d$ and $p \ge 1$, $\lp{x}$ denotes the $\ell_p$ norm of $x$ defined as $(\sum_{i=1}^d |x_i|^p)^{1/p}$ and $\linf{x} \coloneqq \max_i |x_i|$. Given a matrix $A$, we use $\frnorm{A}$ to denote the Frobenius norm and $\|A\|_{p,2}$ to denote the $\ell_{p}$ norm of the $n$-dimensional vector $(\opnorm{a_1},\ldots, \opnorm{a_n})$. Given a matrix $A$, we use $[A]_k$ to denote the best rank-$k$ approximation of $A$ in Frobenius norm. This can be obtained by truncating the singular value decomposition of $A$ to the top $k$ singular values.

For an arbitrary $k$-dimensional subspace $V \in \R^d$, we use $\Proj_V$ to denote the orthogonal projection matrix onto the subspace $V$, i.e., for any $x \in \R^d$, $\Proj_V \cdot x$ is the closest (in Euclidean norm) vector to $x$ in $V$. So, $d(x, V) = \opnorm{(I-\Proj_V)x}$ and $\|A(I-\Proj_V)\|_{\infty, 2} = \max_i \opnorm{(I-\Proj_V)a_i} = \max_i d(a_i, V)$. 
\section{\texorpdfstring{$\ell_\infty$}{l-infty} low rank approximation and Outer Radius}\label{sec:linf-LRA}
\input{linf_LRA.tex}

\input{lp_lra}
\input{applications}
\input{experiments.tex}


\bibliographystyle{plainnat}
\bibliography{main}
\clearpage
\appendix
\input{appendix_neurips}

\end{document}

%% file: linf_LRA.tex
As discussed in the introduction, given a matrix $A$ with rows $a_1, \ldots, a_n$ that arrive in a stream, we want to compute a \emph{strong coreset}, i.e., a subset $S \subseteq [n]$ such that for all $k$-dimensional subspaces $V$,
\begin{align*}
    1 \le \frac{\max_{i \in [n]} d(a_i, V)}{\max_{i \in S}d(a_i, V)} \le f 
\end{align*}
for a small \emph{distortion} $f$. Consider the following simple algorithm: we initialize $S \gets \emptyset$ and stream through the rows $a_1,\ldots,a_n$. When processing the row $a_i$, if there exists a $k$-dimensional subspace $V$ such that $d(a_i, V)^2 > \sum_{i \in S}d(a_i, V)^2$, we update $S \gets S \cup \set{i}$. Otherwise, we proceed to the next row without updating $S$. Consider the set $S$ at the end of the stream and let $V$ be an arbitrary $k$ dimensional subspace. We shall now argue that $A_S$ is a strong coreset with a distortion at most $\sqrt{|S|}$.

 Let $V$ be an arbitrary $k$-dimensional subspace of $\R^d$. Let $i^* = \argmax_i d(a_i, V)$ be the index of the row \emph{farthest} from $V$. Consider the following cases: if $i^* \in S$, then we have $\max_{i \in [n]}d(a_{i}, V) = d(a_{i^*}, V) = \max_{i \in S} d(a_{i}, V)$ and therefore $A_S$ has \emph{no distortion} for $V$. In case the index $i^* \notin S$, then $d(a_{i^*}, V)^2 \le \sum_{i \in S, i < i^*}d(a_{i}, V)^2$ since otherwise we would have added $i^*$ to $S$. Thus, \begin{align}
     \max_i d(a_i, V) = d(a_{i^*}, V) &\le \sqrt{\sum_{i \in S}d(a_i, V)^2}\label{eqn:inf-by-fro}\\
     &\le \sqrt{|S|}\max_{i \in S}d(a_i, V)\nonumber
 \end{align} and therefore $A_S$ is a strong coreset with a distortion at most $\sqrt{|S|}$.
 Now, if we can show that $S$ can not be too large, we obtain that $A_S$ is a strong coreset with a small distortion. 

To show that $S$ is not too large, we appeal to rank-$k$ \emph{online} ridge leverage scores, a generalization of the so-called \emph{ridge leverage scores}. In the offline setting, ridge leverage scores have been employed by \citet{ridge-leverage-score} as a suitable modification of the usual $\ell_2$-leverage scores to obtain fast algorithms for $\ell_2$ low rank approximation. Later, \citet{braverman2020near} defined online ridge leverage scores and showed that they can be used to compute low rank approximations in the \emph{online} model. They also showed that for well-conditioned instances, the sum of the online ridge leverage scores is small. Our main observation is that for the set $S$ constructed as described, the online rank-$k$ ridge leverage score of \emph{every} row in $A_S$ is large. As the sum of online rank-$k$ ridge leverage scores is not large, we obtain that there cannot be too many rows in $A_S$.

One issue we have to solve to implement this algorithm is given $a_i$ and the set $S$ after processing $a_1, \ldots, a_{i-1}$, how can we efficiently know if there exists a rank-$k$ subspace $V$ such that $d(a_i, V)^2 > \sum_{i \in S} d(a_i, V)^2$? Online ridge leverage scores again come to rescue. We show that if we modify the above described algorithm to instead add $i$ to $S$ when its ``online rank-$k$ ridge leverage score'' is large with respect to $A_S$, then the set $S$ computed at the end of the process is again a strong coreset with a distortion of at most $\sqrt{|S|}$.
\subsection{Online Rank-\texorpdfstring{$k$}{k} Ridge Leverage Scores}
Let $A$ be an arbitrary matrix with rows $a_1, \ldots, a_n \in \R^d$ and let $k \le d$ be a rank parameter. Let $\lambda_i = \frac{\frnorm{A_{1:i} - [A_{1:i}]_k}^2}{k}$ be the $i$-th ridge parameter. Note that $\lambda_i = 0$ if and only if $\text{rank}(A_{1:i}) \le k$. We define the ``rank-$k$ online ridge leverage score'' of the row $a_{i+1}$ to be
\begin{align*}
    \tau_{i+1}^{\OL, k}(A) = \begin{cases}
    1\ \text{if}\ \lambda_{i} = 0\ \text{and}\ a_{i+1} \notin \text{rowspace}(A_{1:i})\\
    \min(1, \T{a_{i+1}}(\T{A_{1:i}}A_{1:i} + \lambda_{i} \cdot I)^{+}a_{i+1})\, \text{o.w.}
    \end{cases}
\end{align*}
The online rank-$k$ ridge leverage scores help us capture the ``rank-$k$ information'' of the matrix $A$ as the rows are revealed.
\subsection{An Efficient Algorithm}
\begin{algorithm}
\KwIn{A matrix $A$ as a stream of rows $a_1,\ldots,a_n \in \R^d$, a rank parameter $k$}
\KwOut{A subset $S \subseteq [n]$}	
\DontPrintSemicolon
$S \gets \emptyset$, $\lambda \gets 0$ \tcp*{Algorithm stores $A_S$}
\For{$t=1,\ldots,n$}{
\uIf{$\lambda = 0$ and $a_t \notin \text{rowspace}(A_S)$}{$S \gets S \cup \set{t}$\;}
\uElseIf{$\T{a_t}(\T{A_S}A_S + \lambda \cdot I)^+a_t \ge 1/(1 + 1/k)$}{
        $S \gets S \cup \set{t}$\;
    }
$\lambda \gets \frnorm{A_S - [A_S]_k}^2/k$\; \tcp{$\lambda$ changes only when $S$ changes}
}
\Return{S}\;
\caption{Minimize Distance to a Subspace}
\label{alg:efficient}
\end{algorithm}
Our full coreset construction algorithm is described in Algorithm~\ref{alg:efficient}. In the algorithm, we select a subset of rows $S$ online in the following way: a new row $a_t$ is added to the set $S$ if the rank-$k$ online ridge leverage score of the row $a_t$ with respect to the matrix $A_{S \cup t}$ is at least $1/(1 + 1/k)$. 

We will first show that the set $S$ computed by the algorithm defines a matrix $A_S$ that is a strong coreset with a distortion at most $\sqrt{|S|}$. Let $S_t \coloneqq S \cap [t]$ be the  subset of rows that have been selected by the algorithm after processing $a_1, \ldots, a_t$ and let $a_{t+1}$ is the row being processed. We prove the following lemma:
\begin{lemma}
    Let $t$ be arbitrary and let $S_t \coloneqq S \cap [t]$ be the subset of rows selected by Algorithm~\ref{alg:efficient} after processing the rows $a_1,\ldots,a_t$. If there exists a rank $k$ subspace $V$ such that
    \[
        d(a_{t+1}, V)^2 \ge \sum_{i \in S_t}d(a_i, V)^2,
    \]
    then the algorithm adds the row $t+1$ to the set $S$ that it maintains.
    \label{lma:large-implies-adds}
\end{lemma}
The above lemma now directly implies the following from our earlier discussion:
\begin{lemma}
    Let $S$ be the set returned by Algorithm~\ref{alg:efficient} after processing the rows $a_1, \ldots, a_n$. For any $k$-dimensional subspace $V$,
    \[
    \max_{i\in S}d(a_i, V) \le \max_{i \in [n]}d(a_i, V) \le \sqrt{|S|} \cdot \max_{i \in S}d(a_i, V).
    \]
\end{lemma}
Thus the set $S$ returned by the algorithm is a \emph{strong} coreset with a distortion bounded by $\sqrt{|S|}$. Hence, if we show that $|S|$ is small, then we obtain the two desired properties of a coreset: (i) the distortion of $A_S$ is small and (ii) the number of rows in $A_S$ is small. 

To bound the size of the set $S$, we use the fact that the online rank-$k$ ridge leverage scores of \emph{all} the rows in the matrix $A_S$ \emph{with respect to} $A_S$ are at least $1/(1 + 1/k)$. Thus, the number of rows in $A_S$ is at most $1 + 1/k$ times the sum of online rank-$k$ ridge leverage scores of the matrix $A_S$. We shall now prove a bound on the sum of online rank-$k$ ridge leverage scores of an arbitrary matrix $B$. The proof of this lemma is similar to that of proof of Lemma~2.11 of \cite{braverman2020near}. First, we define an ``online rank-$k$ condition number'' that we use to bound the sum of online rank-$k$ ridge leverage scores.
\begin{definition}[Online Rank-$k$ Condition Number] Given a matrix $B$ with rows $b_1, \ldots, b_n$, let $i^*$ be the largest index $i$ such that $\text{rank}(B_{1:i}) = k$. The online rank-$k$ condition number of $B$ is defined as
\[
    \kappa \coloneqq \frac{\opnorm{B}}{\min_{i \le i^*+1}\sigma_{\min}(B_{1:i})}
\]
where $\sigma_{\min}(\cdot)$ denotes the smallest \emph{non-zero} singular value.

\end{definition}
\begin{lemma}[Sum of online rank-$k$ ridge leverage scores]
    Let $B \in \R^{n \times d}$ be an arbitrary matrix with with an online rank-$k$ condition number $\kappa$, then
    \[\sum_{i=1}^n \tau_i^{\OL, k}(B) = O(k \log(k \cdot \kappa)^2).\]
    \label{lma:sum-of-ridge-leverage-scores}
\end{lemma}
Applying the above lemma to the matrix $A_S$, we obtain that $|S| = O(k \cdot \log(k \cdot \kappa(A_S))^2)$. Using the strong coreset property of the matrix $A_S$, we can show that $\kappa(A_S) \le \sqrt{n} \cdot \kappa(A)$, thereby showing that the coreset has a size at most $|S| = O(k \log(n \cdot \kappa(A))^2)$ and has a distortion at most $O(\sqrt{k}\log(n \cdot \kappa(A)))$. This gives the following theorem:
\begin{theorem}
    Given rows of an arbitrary $n \times d$ matrix $A$ with an online rank-$k$ condition number $\kappa$, Algorithm~\ref{alg:efficient} selects a subset $S$ of size $|S| \le O(k (\log n\kappa)^2)$ such that for any $k$ dimensional subspace $V$, we have
    \begin{align*}
         1 \le \frac{\max_{i \in [n]}d(a_i, V)}{\max_{i \in S} d(a_i, V)}
        \le C\sqrt{k} \cdot \log(n \kappa)
    \end{align*}
    for a large enough constant $C$. Additionally, the space required of the algorithm is bounded by the amount of space required to store $O(|S|)$ rows of $A$.
    \label{thm:first-version}
\end{theorem}
If we assume that all the rows of $A$ lie in a Euclidean ball of radius $R$ and that we are given some $\delta < \Delta \coloneqq \min_{k\text{-dim }V}\max_i d(a_i, V)$, then we can obtain bounds on $|S|$ that are independent of $n$ and only depend on the ``aspect ratio'' $R/\delta$. A similar aspect ratio has been used in an earlier work of \citet{makarychev2022streaming}. Let $t$ be a parameter we fix later. We simply feed the vectors $(\delta/t)e_1, \ldots, (\delta/t)e_{k+1}$ to Algorithm~\ref{alg:efficient} before processing the vectors $a_1, \ldots, a_n$. We note that the algorithm is guaranteed to select the vectors $(\delta/t)e_1, \ldots, (\delta/t)e_{k+1}$ since each of these vectors do not lie in the rowspan of the previous vectors. Let $S$ denote the subset of rows of $A$ selected by this algorithm. Using \eqref{eqn:inf-by-fro}, we note that for any $k$-dimensional subspace $V$,
\begin{align*}
    &\max_{i \in [n]} d(a_i, V) + \max_{i \in [k+1]}d((\delta/t)e_i, V)\\
    &\le \sqrt{\sum_{i=1}^{k+1}d((\delta/t)e_i, V)^2 + \sum_{i \in S}d(a_i, V)^2}
\end{align*}
which implies that
\[
\max_{i \in [n]} d(a_i, V) \le \sqrt{k+1}\frac{\delta}{t} + \sqrt{\sum_{i \in S}d(a_i, V)^2}.
\]
We now note that the online rank-$k$ condition number of the coreset computed by the algorithm must be bounded by $Rt/\delta$ since the first $k+1$ rows of the coreset are guaranteed to be $(\delta/t)e_1, \ldots, (\delta/t)e_{k+1}$. Thus, using Lemma~\ref{lma:sum-of-ridge-leverage-scores} we obtain $|S| = O(k \log (t|S|R/\delta)^2)$, which implies $|S| \le O(k \log (kt \cdot R/\delta)^3)$. If we pick $t = 2\sqrt{k+1}$, we obtain the following theorem.
\begin{theorem}
    Given that $\delta < \max_{k\text{-dim }V}\max_i d(a_i, V)$ and $\opnorm{a_i} < R$, we can compute a subset of rows $A_S$ of $A$ such that for any $k$-dimensional subspace $V$,
    \begin{align*}
        \max_i d(a_i, V) \le C\sqrt{k}(\log kR/\delta)^{3/2}\max_{i \in S}d(a_i, V)
    \end{align*}
    and $|S| = O(k \cdot (\log kR/\delta)^3)$. The space required of the algorithm is bounded by the amount of space required to store $O(|S|)$ rows of the matrix $A$.
    \label{thm:second-version}
\end{theorem}
A coreset $S$ of size $|S|$ and a distortion $\beta$ can also be used to quickly compute an approximate solution to the $\ell_{\infty}$ subspace approximation problem as follows. Let $V^*$ be the optimal solution for the $\ell_{\infty}$ subspace approximation problem on $A$ and $\tilde{V}$ denote the top-$k$ singular subspace of the coreset $A_S$, which can be computed using the singular value decomposition. Then,
\begin{align*}
    \max_i d(a_i, \tilde{V}) \le \beta \cdot \max_{i \in S} d(a_i, \tilde{V}) \le \beta \sqrt{\sum_{i \in S}d(a_i, \tilde{V})^2}.
\end{align*}
Since, $\tilde{V}$ is the top-$k$ singular subspace of the coreset $A_S$, we have 
$
    \sqrt{\sum_{i \in S}d(a_i, \tilde{V})^2} \le \sqrt{\sum_{i \in S}d(a_i, V^*)^2} 
$
which overall implies
\begin{align*}
    \max_i d(a_i, \tilde{V}) &\le \beta \sqrt{\sum_{i \in S}d(a_i, V^*)^2} 
    \le \beta\sqrt{|S|}\max_i d(a_i, V^*).
\end{align*}    
Hence, a $\beta\sqrt{|S|}$ approximation to the $\ell_{\infty}$ subspace approximation\footnote{In our case, the approximation factor is $O(k (\log n\kappa)^2)$.} problem can be obtained without using any SDP based algorithms from previous works. We can additionally initialize an alternating minimization algorithm on the coreset for $\ell_{\infty}$ subspace approximation using the SVD subspace of the coreset and use convex optimization solvers to further improve the quality of the solution. We do note that there are no known bounds on the solution quality attained by the alternating minimization algorithm. 

By a simple (lossy) reduction of the outer $(d-k)$ radius estimation problem to computing optimal $\ell_{\infty}$ subspace approximation of the matrix $B = A - a_1$,  i.e., the matrix obtained by subtracting $a_1$ from each row of $A$, we obtain the following theorem using the coreset bounds in Theorem~\ref{thm:second-version}.
\begin{theorem}[Outer $(d-k)$ radius estimation]
    Given $0 = a_1 - a_1, \ldots, a_n - a_1$, if a streaming algorithm computes a coreset $S$ with distortion $\beta$, then the outer $(d-k)$ radius of the point set $S$ is an $O(\beta)$ approximation to the outer $(d-k)$ radius of the entire point set.

Given that the online rank-$k$ condition number of the matrix $A-a_1$ is $\kappa'$, the outer $(d-k)$ radius of the point set can be approximated up to a $\sqrt{k} \cdot \log n\kappa'$ factor by computing the outer $(d-k)$ radius of the coreset points.
    \label{thm:outer-radius}
\end{theorem}
\subsection{Fast Implementation of Algorithm~\ref{alg:efficient}}
Note that the set $S$ and hence the value $\lambda$ are updated only at most $O(k\log(n \cdot \kappa)^2)$ times in the stream. Hence, if we compute the singular value decomposition of $A_S$ each time $S$ is updated, we only spend at most $O(d \poly(k, \log n\kappa))$ time in total. Let $U\Sigma\T{V} = A_S$ be the ``thin'' singular value decomposition of $A_S$. Then given any vector $a$, we can compute $\T{a}(\T{A_S}A_S + \lambda I)^{+}a$ as $\opnorm{\Sigma^{-1}\T{V}a}^2 + (1/\lambda)\opnorm{(I-V\T{V})a}^2 = \opnorm{Ma}^2$ where $M$ is defined as the matrix obtained by concatenating $\Sigma^{-1}\T{V}$ and $(1/\sqrt{\lambda})(I-V\T{V})$. 

Now, if $\bG$ is a Gaussian matrix with $O(\log n)$ rows, we can approximate $\opnorm{Ma_i}^2$ with $\opnorm{\bG M a_i}^2$ up to constant factors for all the \emph{future} rows $a_i$. Suppose each time $S$ is updated, we compute the matrix $M$ and sample a Gaussian matrix $\bG$ and then compute $\bG M$ which has $O(\log n)$ rows. Then the online rank-k ridge leverage score of any row $a_i$ that appears in the stream can be approximated as $\opnorm{(\bG M)a_i}^2$ in time $O(\text{nnz}(a_i) \log n)$, since the matrix $\bG M$ has only $O(\log n)$ rows. Thus the overall algorithm can be implemented in time $O(\nnz(A)\log n + d \cdot \poly(k, \log n\kappa))$. We implement this algorithm and find that it runs very fast on large datasets.
\section{Lower Bounds}
The algorithm in the previous section uses $O(dk (\log n\kappa)^2)$ bits of space to process a stream of $n$ rows in $\R^d$ and outputs a strong coreset with a distortion at most $O(C\sqrt{k}\log n\kappa)$, where $\kappa$ is the condition number. We show that any algorithm that constructs a strong coreset with distortion $O(\sqrt{k/\log n})$ must use $\Omega(n)$ bits of space. This shows that our algorithm obtains the best possible distortion bounds up to $\poly(\log n\kappa)$ factors. Our argument is similar to that of \cite{wy22}. We state the lower bound in the following theorem.
\begin{theorem}
Given parameters $n$, $d$ and $k$ with $k = \Omega(\log n)$, any streaming algorithm that computes a strong coreset with distortion at most $O(\sqrt{k/\log n})$ with probability $\ge 9/10$ must use $\Omega(n)$ bits of space.
\end{theorem}
\begin{proof}
    Let $n$, $d$ and $k$ be arbitrary. Let $a_1, \ldots, a_{2n} \in \R^{d}$ be random vectors sampled as follows: each of the first $k$ entries of each $a_i$ is set to $+1/-1$ with equal probability. The remaining $d-k$ coordinates of each $a_i$ are set to $0$.

    Note that $\opnorm{a_i}^2 = k$ for all $i$. For arbitrary $i \ne j$, consider $|\la a_i, a_j\ra|$. By Hoeffding's inequality, with probability $\ge 1 - \delta$, $|\la a_i, a_j\ra| \le O(\sqrt{k\log 1/\delta})$. Setting $\delta = 1/10n^2$ and using a union bound, we obtain that with probability $\ge 9/10$, for all $i \ne j$, $|\la a_i, a_j\ra| \le O(\sqrt{k\log n})$. Condition on this event. Let $\bS \subseteq [2n]$, $|\bS| = n$ be a uniformly random subset of $[2n]$ of size $n$.

    Consider the stream of vectors $(a_i)_{i \in \bS}$. Let $\calC$ be a randomized algorithm that computes a strong coreset with distortion $\alpha \le O(\sqrt{k/\log n})$ with probability $\ge 9/10$. Let $\calC((a_i)_{i \in \bS})$ be the output of the  algorithm $\calC$ on the stream $(a_i)_{i \in \bS}$. Condition on the event that $\calC((a_i)_{i \in \bS})$ is a strong coreset. We now argue that if $\alpha$ is not too large, we can compute the set $\bS$ from the coreset $\calC((a_i)_{i \in \bS})$. 

    Given a strong coreset $M$ with distortion $\alpha$ for the stream $(a_i)_{i \in \bS}$, and a rank-$k$ subspace $V$, let $M(V)$ be the value computed using the coreset such that
    \begin{align*}
        M(V) \le \max_{i \in \bS}d(a_i, V) \le \alpha \cdot M(V).
    \end{align*}
    For each $i \in [2n]$, consider the subspace $V_i = \text{span}(e_1, \ldots, e_k)\, \cap\, a_i^{\perp}$,
    where $a_i^{\perp}$ denotes the subspace orthogonal to the vector $a_i$. We now note the following:
    \vspace{-1em}
    \begin{itemize}\itemsep 0em
        \item $d(a_i, V_i) = \opnorm{a_i} = \sqrt{k}$
        \item For all $j \ne i$, $d(a_j, V_i) = |\la a_j, a_i\ra|/\opnorm{a_i} \le O(\sqrt{\log n}).$
    \end{itemize}
    \vspace{-1em}
    Therefore if $i \in \bS$, then $\calC((a_j)_{j \in \bS})(V_i) \ge \sqrt{k}/\alpha$ and if $i \notin \bS$, then $\calC((a_j)_{j \in \bS})(V_i) \le O(\sqrt{\log n})$. If the distortion $\alpha \le \sqrt{k/\log n}$, then by enumerating over all $V_i$ for $i \in [2n]$ and computing $\calC((a_j)_{j \in \bS})(V_i)$, we can determine the set $\bS$.

    Let $\bS'$ be the set computed by the enumeration algorithm. If $|\bS'| \ne n$, set $\bS'$ to $\set{1, 2, \ldots, n}$. By the above discussion, we have $\Pr[\bS' = \bS] \ge 9/10$. Note that the entropy of the set $\bS$ is $t = \Omega(n)$ where $2^t = \binom{2n}{n}$ is the number of subsets of $[2n]$ of size $[n]$.
    
    We now upper bound the conditional entropy  $H(\bS' \mid \bS)$. Let $\bI$ denote the indicator random variable denoting if the coreset construction algorithm succeeds. Note that given $\bI = 1$, we have $\bS = \bS'$. We have $H((\bS, \bS')) = H(\bS) + I(\bS\, ;\, \bS')$ and 
    \begin{align*}
        H((\bS, \bS')) &\le H((\bS, \bS', \bI)) = H(\bS) + H(\bI \mid \bS) + H(\bS' \mid \bI, \bS)
    \end{align*}
    and therefore, $I(\bS\, ;\, \bS') \le H(\bI \mid \bS) + H(\bS' \mid \bI, \bS)$. Since we assumed that the coreset construction algorithm succeeds with probability $\ge 9/10$ given any instance, we have $H(\bI \mid \bS) \le (9/10)\log_2(10/9) + (1/10)\log_2(10) \le 1/2$. Now,
    \begin{align*}
        &H(\bS' \mid \bI, \bS)\\
        &= \sum_{S} \Pr[\bS = S] \cdot \left[ H(\bS' \mid \bS = S, \bI = 0) \cdot \Pr[\bI = 0 \mid \bS = S]\right. \\
        &\qquad \left.+ H(\bS' \mid \bS = S, \bI = 1) \cdot \Pr[\bI = 1 \mid \bS = S]\right]\\
        &\le \sum_{S} \Pr[\bS = S] \cdot H(\bS'\ \mid \bS = S, \bI = 0) \cdot (1/10)
    \end{align*}
    
    where we used the fact that if $\bI = 1$, then $\bS' = \bS$ and therefore $H(\bS' \mid \bS = S, \bI = 1) = 0$. Since the output $\bS'$ is always a subset of $[2n]$ of size $n$, we have $H(\bS' \mid \bS = S, \bI = 0) \le \log_2\binom{2n}{n} = t$ which then implies
$
        H(\bS' \mid I, \bS) \le t/10.
 $
 Hence the mutual information $I(\bS\, ;\, \bS') \ge 9t/10 - 1/2$ and by the data processing inequality, we have
    \begin{align}
        I(\calC((a_i)_{i \in \bS})\, ;\, \bS) \ge 9t/10 - 1/2
    \end{align}
    which implies that the space necessary to store the coreset is $\Omega(n)$ bits since $t = \log_2 \binom{2n}{n} = \Omega(n)$.
\end{proof}

%% file: lp_lra.tex
\section{\texorpdfstring{$\ell_p$}{lp} Subspace Approximation}\label{sec:lp-lra}
\vspace{-0.5em}
We now show that our coreset construction algorithm for the $\ell_{\infty}$ subspace approximation problem, extends to the $\ell_p$ subspace approximation problem. Fix a matrix $A$. For any $k$-dimensional subspace $V$, let $d_V$ denote the non-negative vector satisfying
$
    (d_V)_i = \dist(a_i, V) = \opnorm{\T{a_i}(I-\Proj_V)}.
$
Hence, the $\ell_p$ subspace approximation problem is to find the rank-$k$ subspace $V$ that minimizes $\lp{d_V}$. We use exponential random variables to embed an $\ell_p$ low rank approximation problem into an $\ell_{\infty}$ low rank approximation problem. We then use the coreset construction algorithm for $\ell_{\infty}$ LRA to obtain a coreset for the $\ell_p$ LRA. First, we have the following lemma about exponential random variables that has been used in various previous works to embed $\ell_p$ problems into an $\ell_{\infty}$ problem.
\begin{lemma}
    Let $\be_1,\ldots,\be_n$ be independent exponential random variables. Then with probability $\ge 1 - \delta$,
    $
        \max_i \be_i^{-1/p}|x_i| \ge {\lp{x}}/{(\log 1/\delta)^{1/p}}.
    $
We also have that with probability $\ge 1 - \delta$,
$
    \max_i \be_i^{-1/p}|x_i| \le \delta^{-1/p} \cdot \lp{x}.
$
\label{lma:exponential}
\end{lemma}
\begin{proof}
By min-stability of exponential random variables, we have that the distribution of $\max_i \be^{-1}|x_i|^p$ is the same as the distribution of $\be^{-1}\lp{x}^p$ where $\be$ is also a standard exponential random variable. With probability $\ge 1 - \delta$, we have $\be \le \log 1/\delta$. And hence we have that with probability $\ge 1 - \delta$,
\begin{align*}
    \max_i \be_i^{-1/p}|x_i| = (\max_i \be_i^{-1}|x_i|^p)^{1/p} \ge \frac{\lp{x}}{(\log 1/\delta)^{1/p}}.
\end{align*}

With probability $\ge 1 - \delta$, we also have that $\be \ge \delta$ which implies that with probability $\ge 1 - \delta$, $\max_i \be_i^{-1/p}|x_i| = (\max_i \be_i^{-1}|x_i|^p)^{1/p} \le \lp{x}\delta^{-1/p}$.
\end{proof}

Given $n$, define $\bD$ to be a random matrix with diagonal entries given by independent copies of the random variable $\be^{-1/p}$. For any fixed rank $k$ projection matrix $P$, the above lemma implies that $\|\bD A(I-P)\|_{\infty, 2} \ge \|A(I-P)\|_{p,2}/(\log 1/\delta)^{1/p}$. However, we cannot union bound over the net of all $k$-dimensional subspaces of $\R^d$ since the net can have as many as $\exp(dk)$ subspaces which leads to a distortion of $d^{1/p}$, which is prohibitive. Here we crucially use the fact that Algorithm~\ref{alg:efficient} only selects a coreset with $m = O(k \cdot (\log n\kappa)^2)$ rows. Thus, only those $k$-dimensional subspaces spanned by at most $m$ rows of $A$ are of interest to us. Now, we can union bound over a net of $\exp(\poly(k, \log n\kappa))$ subspaces and show the following lemma:
\begin{lemma}
    Let $\bD$ be an $n \times n$ diagonal matrix with each diagonal entry being an independent copy of the random variable $\ceil{\be^{-1/p}}$. Fix an $n \times d$ matrix $A$. With probability $\ge 98/100$, for all $k$-dimensional subspaces that are in the span of at most $m = O(k \log^2 n\kappa)$ rows of $A$, we have,
\begin{align*}
    \linf{\bD \cdot d_V} \ge \frac{\lp{d_V}}{2(\log 100 + m \log n + km \log n\kappa)^{1/p}}. 
\end{align*}
\label{lma:net-argument}
\end{lemma}
\begin{proof}
Let $S$ be an arbitrary set of $m \le K$ rows of $A$ and let $V_S \coloneqq \text{rowspace}(A_S)$. Let $N_S$ be a $\gamma$ net for the set $V_S \cap \mathbb{S}^{d-1}$ i.e., the set of vectors in the subspace $V_S$ with euclidean norm $1$. As the subspace $V_S$ has dimension at most $m$, we have that there is a set $N_S$ with size at most $\exp(O(m\log 1/\gamma))$. Let $V$ be an arbitrary $k$ dimensional subspace of $V_S$ and let $\set{v_1,\ldots,v_k}$ be an orthonormal basis for $V$. 

Let $\tilde{V}$ be the subspace spanned by $\set{\tilde{v}_1,\ldots,\tilde{v}_k}$, where $\tilde{v}_i \in N_S$ and $\opnorm{v_i - \tilde{v}_i} < \gamma$ for all $i \in [n]$. Let $a$ be an arbitrary vector. By abusing the notation let $V$ (resp. $\tilde{V}$) also denote the matrix with $v_1,\ldots, v_k$ (resp. $\tilde{v}_1, \ldots, \tilde{v}_k$) as columns. We have
\begin{align*}
    d(a, V) = \opnorm{a - V\T{V}a}\quad \text{and}\quad d(a, \tilde{V}) = \opnorm{a - \tilde{V}\tilde{V}^+ a}
\end{align*}
and therefore
$
    |d(a, V) - d(a, \tilde{V})| \le \opnorm{\tilde{V}\tilde{V}^+ - V\T{V}}\opnorm{a}.
$
If $\gamma \le 1/4\sqrt{k}$, we can show that $\opnorm{V\T{V} - \tilde{V}\tilde{V}^+} \le 4\sqrt{k}\gamma$ and therefore have that for any $a$, $|d(a, V) - d(a, \tilde{V})| \le \sqrt{k}\gamma\opnorm{a}$. Hence,
\begin{align*}
    \linf{d_V - d_{\tilde{V}}} \le \max_i |d(a_i, V) - d(a_i, \tilde{V})| \le 4\sqrt{k}\gamma\max_i \opnorm{a_i} = 4\sqrt{k}\gamma\|A\|_{\infty, 2}.
\end{align*}
Overall, this implies that for any arbitrary $k$ dimensional subspace $V$ in the span of rows of $A_S$, there is a $k$ dimensional subspace $\tilde{V}$ spanned by some $k$ vectors in the net $N_S$ satisfying
\begin{align*}
	\linf{d_V - d_{\tilde{V}}} \le 4\sqrt{k}\gamma\|A\|_{\infty,2}.
\end{align*}
As $d_V \in \R^n$, we have $\lp{d_V - d_{\tilde{V}}} \le n^{1/p}\linf{d_V - d_{\tilde{V}}} \le 4\sqrt{k}\gamma n^{1/p}\|A\|_{\infty, 2}$. Now, let
\begin{equation*}
	\calV_S := \setbuilder{\tilde{V} = \text{span}(\tilde{v}_1,\ldots,\tilde{v}_k)}{\tilde{v}_i \in N_S}.
\end{equation*}
We have $|\calV_S| \le |N_S|^k \le \exp(O(km\log 1/\gamma))$ since $|N_s| \le \exp(O(m\log 1/\gamma))$. As there are $\binom{n}{m}$ choices for $S$, the total number of subspaces in the set $\cup_{S \in \binom{[n]}{m}}\calV_S$ is upper bounded by $\exp(m\log n + km\log 1/\gamma)$. Using Lemma~\ref{lma:exponential}, using a union bound over all $\exp(m \log n + km\log 1/\gamma)$ choices of $\tilde{V}$, we have that with probability $\ge 99/100$, for all $\tilde{V} \in \cup_{\binom{[n]}{m}} V_S$,
\begin{align*}
	\linf{\bD \cdot d_{\tilde{V}}} \ge \frac{\lp{d_{\tilde{V}}}}{(\log 100 + m\log n + km\log 1/\gamma)^{1/p}}.
\end{align*}
Using Lemma~\ref{lma:exponential} again, we also have that $\max_i |\bD_i| \le C_3n^{1/p}$ for a large enough constant $C_3$ with probability $\ge 99/100$. Condition on both these events. We have that for any $k$ dimensional subspace $V$ in the span of any set of $m$ rows of $A$,
\begin{align*}
	\linf{\bD \cdot d_V} &\ge \linf{\bD \cdot d_{\tilde{V}}} - \linf{\bD \cdot (d_V - d_{\tilde{V}})}\\
		&\ge \frac{\lp{d_{\tilde{V}}}}{(\log 100 + m\log n + km\log 1/\gamma)^{1/p}} - C_1n^{1/p}\linf{d_V - d_{\tilde{V}}}\\
		&\ge \frac{\lp{d_{{V}}}}{(\log 100 + m\log n + km\log 1/\gamma)^{1/p}} - \frac{4\sqrt{k}n^{1/p}\gamma \|A\|_{\infty,2}}{(\log 100 + m\log n + km\log 1/\gamma)^{1/p}}\\
		&\quad - 4C_1n^{1/p}\sqrt{k}\gamma\|A\|_{\infty, 2}.
\end{align*}
For any $V$, we have that $\lp{d_V} \ge \opnorm{d_V}/\sqrt{n} \ge \frnorm{A - [A]_k}/\sqrt{n}$ using the fact that $V$ is a $k$ dimensional subspace. Hence, if $\gamma \le \poly(\frnorm{A-[A]_k}/\|A\|_{\infty,2}, 1/n)$, then
\begin{align*}
	\linf{\bD \cdot d_V} \ge \frac{\lp{d_V}}{2(\log 100 + m\log n + km\log 1/\gamma)^{1/p}}.
\end{align*}
Now, $\gamma$ can be taken as $\poly(1/(n\kappa))$ so that
\begin{align*}
    \linf{\bD \cdot d_{V}} \ge \frac{\lp{d_V}}{C(\log 100 + m \log n + km\log(n\kappa))^{1/p}}
\end{align*}
for all subspaces $V$ that are in the span of any subset of $m$ rows of $A$.
\end{proof}

If $V^*$ is the optimal solution for the $\ell_p$ subspace approximation problem, we can also condition on the event that $\linf{\bD \cdot d_{V^*}} \le C\lp{d_{V^*}}$ for a large enough constant $C$.

We can now argue that if $S$ is the subset of rows selected by Algorithm~\ref{alg:efficient} when run on the matrix $\bD A$, if $\hat{V}$ is an approximate solution for the $\ell_{\infty}$ subspace approximation problem on the points $(\bD A)_S$, then $\hat{V}$ is also a good solution for the $\ell_p$ subspace approximation problem of $A$.
\begin{theorem}
    Let $\bD$ be an $n \times n$ random matrix with each diagonal entry being an independent copy of $\ceil{\be^{-1/p}}$ where $\be$ is a standard exponential random variable. If $S$ is the subset selected by Algorithm~\ref{alg:efficient} when run on the rows of the matrix $\bD \cdot A$ and if $\hat{V}$ is a $\beta$ approximate solution to the problem $\min_{k\text{-dim }V}\|(\bD A)_S (I-\Proj_V)\|_{\infty, 2}$, then with probability $\ge 9/10$,
    \begin{align*}
        \frac{\|A(I-\Proj_{\hat{V}})\|_{p,2}}{\min_{k\text{-dim }V}\|A(I-\Proj_{V})\|_{p,2}} \le \beta \cdot O(k^{1/2 + 2/p}\log^{1 + 3/p}n\kappa).
    \end{align*}
\end{theorem}
\begin{proof}
    Let 
\begin{align*}
	V^* = \argmin_{k\text{-dim subspaces}\ V}\lp{d_V}.
\end{align*}
Condition on the event that $\linf{\bD^{1/p} \cdot d_{V^*}}\le C_1\lp{d_{V^*}}$ for a large enough constant $C_1$. The event holds with probability $\ge 99/100$ by Lemma~\ref{lma:exponential}. Finally, by a union bound, we have all the following events hold simultaneously with probability $\ge 9/10$:
\begin{enumerate}
	\item Algorithm~\ref{alg:efficient}, when run on the rows of the matrix $\bD \cdot A$, selects at most $m = O(k \cdot (\log n\kappa)^2)$ rows.
	\item For any $k$ dimensional subspace $V$ contained in the span of any at most $m$ rows of $A$,
	\begin{align*}
		\linf{\bD \cdot d_V} \ge \frac{\lp{d_V}}{C_2k^{2/p}\log^{3/p}n\kappa}.
	\end{align*}
	\item If $V^*$ is the optimal subspace that minimizes the $\ell_p$ norm of the distance vector to a $k$ dimensional subspace, then
	\begin{align*}
		\linf{\bD \cdot d_{V^*}} \le C_1\lp{d_{V^*}}.
	\end{align*}
\end{enumerate}
Conditioned on the above events, let $S \subseteq [n]$ be the coreset computed for the matrix $\bD \cdot A$ by Algorithm~\ref{alg:efficient}. From Theorem~\ref{thm:first-version}, we have that for any rank $k$ projection matrix $P$, 
\begin{align*}
	\|{(\bD A)_S(I-P)}\|_{\infty,2} \le \|(\bD A)(I-P)\|_{\infty,2} \le C\sqrt{k}(\log n\kappa)\|{(\bD A)_S(I-P)}\|_{\infty,2}.
\end{align*}
Let $\hat{V}$ be a $k$ dimensional subspace such that
\begin{align*}
   \min_{k\text{-dim }V}\|{(\bD A)_S(I-\Proj_{\hat{V}})}\|_{\infty,2} \beta \cdot \min_{k\text{-dim }V}\|{(\bD A)_S(I-\Proj_V)}\|_{\infty,2}
\end{align*}
Without loss of generality, we can assume that $\hat{V}$ is contained in the rowspace of $(\bD \cdot A)_S$ and hence the row space of $A_S$. Therefore,
\begin{align*}
	\|A(I-\Proj_{\hat V})\|_{p,2} &= \lp{d_{\hat{V}}}\\
	&\le C_2k^{2/p}\log^{3/p}(n\kappa)\linf{\bD \cdot  d_{\hat{V}}}\\
	&= C_2k^{2/p}\log^{3/p}(n\kappa)\|(\bD \cdot A)(I-\Proj_{\hat V})\|_{\infty,2}\\
	&\le C_2 \cdot C \cdot  k^{2/p+1/2}\log^{1+3/p}(n\kappa)\|(\bD A)_S(I-\Proj_{\hat V})\|_{\infty,2}\\
	&\le \beta \cdot C_2 \cdot C \cdot  k^{2/p+1/2}\log^{1+3/p}(n\kappa)\|(\bD A)_S(I-\Proj_{V^*})\|_{\infty,2}\\
	&= \beta \cdot C_2 \cdot C \cdot  k^{2/p+1/2}\log^{1+3/p}(n\kappa)\|\bD \cdot d_{V^*}\|_{\infty,2}\\
	&\le \beta \cdot C_1 \cdot C_2 \cdot C \cdot k^{2/p + 1/2}\log^{1+3/p} (n\kappa)\lp{d_{V^*}}.
\end{align*}
Thus, $\hat{V}$ is an $O(\beta \cdot k^{2/p+1/2}\log^{1+3/p}(n\kappa))$ approximate solution for the $\ell_p$ low rank approximation problem over the matrix $A$.
\end{proof}

%% file: applications.tex
\section{Applications to Other Geometric Streaming Problems}
Given a matrix $A$, suppose that the rows of $A$ are close to a $k$-dimensional subspace in the following sense: $\Delta \coloneqq \min_{k\text{-dim }V}\max_i d(a_i, V)$ is small.
We now show that if $S$ is the subset of rows selected by Algorithm~\ref{alg:efficient}, then for any vector $x$, $\linf{Ax}$ can be approximated using $\linf{A_S x}$. Fix any unit vector $x$. Let $i$ be the index such that $\linf{Ax} = |\la a_i, x\ra|$. If $i \in S$, we clearly have $\linf{Ax} = \linf{A_S x}$ and we  are done. If $i \notin S$, we obtain that
\begin{align*}
    \max_x \frac{|\la a_i, x\ra|^2}{\opnorm{A_{S < i}x}^2 + \frnorm{A_{S < i} - [A_{S < i}]_k}^2/k} \le \frac{1}{1+1/k}
\end{align*}
which implies
\begin{align*}
    \linf{Ax}^2 = |\la a_i, x\ra|^2 &\le \opnorm{A_{S < i}x}^2 + \frac{\frnorm{A_{S < i} - [A_{S < i}]_k}^2}{k}\\
    &\le \opnorm{A_S x}^2 + \frac{\frnorm{A_S - [A_S]_k}^2}{k}.
\end{align*}
Let $V^*$ be the optimal solution for rank-$k$ $\ell_{\infty}$ subspace approximation of $A$.
We then have,
$
    \linf{Ax}^2 \le \opnorm{A_S x}^2 + {\frnorm{A_S(I-\Proj_{V^*})}^2}/{k} \le \opnorm{A_S x}^2 + {|S|\Delta^2}/{k}.
$
Using $|S| = O(k \log^2 n\kappa)$, we get the following lemma.
\begin{lemma}
If $S$ is the subset of rows selected by Algorithm~\ref{alg:efficient}, for any $k$-dimensional subspace $U$ and any unit vector $x$,
   \begin{small}
    \begin{align*}
    \frac{\opnorm{A_S x}}{C\sqrt{k}\log n\kappa}\le \linf{A_S x} \le \linf{Ax} \le \opnorm{A_S x} + C\Delta\log n\kappa.
\end{align*}
\end{small}
Additionally, as $\opnorm{A_S x} \le \sqrt{|S|}\linf{A_S x}$, we also have
\begin{small}
    \begin{align*}
    \linf{A_S x} \le \linf{Ax} \le (C\sqrt{k}\log n\kappa)\linf{A_S x} + C\Delta\log n\kappa.
\end{align*}
\end{small}
\label{lma:linf-Ax-estimation}
\end{lemma}
\vspace{-2em}
\textbf{Width Estimation.} Given a point set $a_1,\ldots,a_n \in \R^d$, the width of the point set in the direction $x \in \R^d$, for a unit vector $x$ is defined as 
$
    w(x) \coloneqq \max_i \la a_i, x\ra - \min_i \la a_i, x\ra.
$
Using a coreset for estimating $\linf{Ax}$, \cite{wy22} gives an $O(\sqrt{d\log n})$ approximation to the width estimation problem. Using Lemma~\ref{lma:linf-Ax-estimation}, we show that we get better approximations when $\Delta$ is small. 

Note that
$
    w(x) = \max_i \la a_i - a_1, x\ra - \min_i \la a_i - a_1, x\ra.
$
Now, $\max_i \la a_i - a_1, x\ra \ge \la 0, x\ra = 0$ and $\min_i \la a_i - a_1, x\ra \le \la 0, x\ra \le 0$ which implies that
$
    \linf{(A - a_1)x} \le w(x) \le 2\linf{(A - a_1)x}. 
$

Let $\kappa'$ be the online rank-$k$ condition number of $A-a_1$. If $S$ is the subset selected by the algorithm when run on the rows $0 = a_1 - a_1, a_2 - a_1, \ldots, a_n - a_1$, then from Lemma~\ref{lma:linf-Ax-estimation}, we have
$
    \linf{(A-a_1)_S x} \le \linf{(A-a_1)x} \le w(x)
$
and also that 
$
    w(x) \le 2\linf{(A-a_1) x} \le 2C\sqrt{k}\log (n\kappa')\linf{(A-a_1)_S x} + 2C\Delta\log (n\kappa'). 
$
Thus, $w'(x) \coloneqq \linf{(A-a_1)_Sx}$ satisfies
\begin{align*}
    {w(x)}/{2C\sqrt{k}\log (n\kappa')} - {\Delta}/{\sqrt{k}} \le w'(x) \le w(x)
\end{align*}
for a large enough constant $C$. When $\Delta$ is very small, for the interesting directions where width is large enough, we obtain a better multiplicative error of $O(\sqrt{k}\log n\kappa')$ as compared to $O(\sqrt{d\log n})$ achieved by the algorithm of \cite{wy22}. Notice that we do not contradict the lower bounds of \citet{agarwal2015streaming} for width estimation because of the additive error that we allow.

\textbf{L\"owner-John Ellipsoid.}
Given a symmetric convex body, the L\"owner-John ellipsoid is defined to be the ellipsoid of minimum volume that encloses the convex body. We consider the case when the convex body is defined as $K = \setbuilder{x}{\linf{Ax} \le 1}$ where the streaming algorithm sees the rows of matrix $A$ one after the other. \citet{wy22} show that their coreset can be used to compute an ellipsoid $E'$ such that $E' \subseteq K \subseteq O(\sqrt{d\log n})E'$.

When $k \ll d$, Algorithm~\ref{alg:efficient} selects $\ll d$ number of rows and does not have the full $d$-dimensional \emph{view} of the point set and hence can not compute an ellipsoid that satisfies the above multiplicative definition if the points span $\R^d$. Thus, we consider the set $K \cap B(0,1)$ and give an algorithm that computes an unbounded ellipsoid $E'$ such that $E' \cap B(0,1) \subseteq K \cap B(0,1)  \subseteq (\alpha E') \cap B(0,1)$.

By Lemma~\ref{lma:linf-Ax-estimation}, we have that if $\linf{Ax} \le 1$ and $\opnorm{x} = 1$, then $\opnorm{A_S x} \le C\sqrt{k}\log n\kappa$ and if $\opnorm{A_S x} \le 1  - C\Delta\log n\kappa$ and $\opnorm{x} \le 1$, then $\linf{Ax} \le 1$. Now assuming $\Delta  < 1/(C\log n\kappa)$, define 
$
    E' = \setbuilder{x}{\opnorm{A_S x} \le 1 - (C\log n\kappa)\Delta}.
$

From the above, we have that if $x \in E' \cap B(0,1)$, then $x \in K \cap B(0,1)$. Additionally if $x \in K \cap B(0,1)$, then $\opnorm{A_S x} \le C\sqrt{k}\log n\kappa$ and therefore $x \in \frac{C\sqrt{k}\log n\kappa}{1 - (C\Delta\log n\kappa)}E' \cap B(0,1)$. Hence,
    \begin{align*}
    E' \cap B(0,1) \subseteq K \cap B(0,1) \subseteq \frac{C\sqrt{k}\log n\kappa}{1 - (C\Delta\log n\kappa)}E' \cap B(0,1).
\end{align*}

%% file: experiments.tex
\vspace{-2em}
\section{Experiments}
We implement our coreset construction algorithm (Algorithm~\ref{alg:efficient}) and show that the coreset has a low distortion both for the $\ell_{\infty}$ low rank approximation problem and for width estimation. 
\subsection{\texorpdfstring{$\ell_{\infty}$}{l-infty} low rank approximation}
 We run Algorithm~\ref{alg:efficient} on a synthetic data set and a real world dataset. We construct our synthetic dataset as follows: we pick $n = 40{,}000$, $d = 10{,}000$ and $k = 20$.  We sample an $n \times k$ random matrix $L$ and a $k \times d$ random matrix $R$ each with i.i.d. uniform random variables drawn  from $\set{-100, -99, \ldots, 100}$. We create an $n \times d$ matrix $A \doteq L \cdot R + G$ where $G$ is a noise matrix with each entry being an i.i.d. uniform random variable drawn from $\set{-5000,\ldots,5000}$. With parameter $k = 20$, when Algorithm~\ref{alg:efficient} is run on the matrix $A$, the coreset $A_S$ computed by the algorithm has only $28$ rows. To measure the \emph{quality} of the coreset, we consider the following candidate subspaces: we define $V_i$ to be the at most $i$-dimensional subspace formed by the first $i$ rows of $R$. These are indeed the subspaces for which the rows of $A$ have a \emph{low} distance to. We obtain that
\begin{align*}
    1 \le \max_{i \in [20]} \frac{\|A(I - \Proj_{V_i})\|_{\infty, 2}}{\|A_S(I-\Proj_{V_i})\|_{\infty, 2}} \le 1.3433
\end{align*}
which shows that the $\ell_{\infty}$ cost of the interesting subspaces estimated using the coreset is not too small compared to the actual $\ell_{\infty}$ cost of the subspace. Another important requirement is that we do not underestimate the cost of uninteresting subspaces by a lot. To see this, we generate random subspaces of $k = 20$ dimensions and observe that $\|A(I-\Proj_{V})\|_{\infty, 2}/\|A_S(I-\Proj_{V})\|_{\infty, 2} \le 1.05$ with high probability when $V$ is drawn at random. This can be explained by the fact that random subspaces are so bad in that $\|A(I-\Proj_V)\|_{\infty, 2} \approx \|A\|_{\infty, 2}$ since a random subspace does not capture a large part of the row of $A$ with the largest norm. We see that when $V$ is a random matrix, $\|A(I-\Proj_V)\|_{\infty, 2}/\|A_S(I-\Proj_V)\|_{\infty, 2} = \|A\|_{\infty, 2}/\|A_S\|_{\infty}$ and since all the rows of $A$ have similar norms, we get that $\|A(I-\Proj_V)\|_{\infty, 2}/\|A_S(I-\Proj_V)\|_{\infty, 2} \approx 1$.

For the real world dataset, we consider a grayscale image \cite{chessboard} of dimensions $1836 \times 3264$ and treat the image as a matrix $A$ of the same dimensions. We observe that a rank-$150$ approximation of the image computed using the SVD is very close to the original image (with some artifacts) and therefore set $k=150$ to be the parameter for which we want to solve the $\ell_\infty$ low rank approximation problem. We run the coreset construction algorithm on $A$ and obtain a coreset $A_S$ with $312$ rows. Note that the number of rows in the coreset is $\approx 17\%$ of the original matrix. Again, to measure the quality of the coreset, we consider subspace $V_i$ defined to be the top $i$-dimensional right singular subspace of $A$ and measure $\|A(I-\Proj_{V_i})\|_{\infty, 2}/\|A_S(I-\Proj_{V_i})\|$. We obtain $\max_{i \in [k]}\|A(I-\Proj_{V_i})\|_{\infty, 2}/\|A_S(I-\Proj_{V_i})\|_{\infty, 2} \le 1.09$ and hence the coreset gives very accurate cost estimates for these interesting subspaces. We repeat the same experiment on a different grayscale image \cite{galaxy} of dimensions $4690 \times 6000$ and use $k = 200$. We obtain a coreset $A_S$ with $382$ rows and for $V_i$ defined in the same way as before, $\max_i \|A(I-\Proj_{V_i})\|_{\infty, 2}/\|A_S(I-\Proj_{V_i})\|_{\infty, 2} \le 1.12$.


\subsection{Width Estimation}
Towards width estimation, Lemma~\ref{lma:linf-Ax-estimation} shows that if $A_S$ is the coreset computed by Algorithm~\ref{alg:efficient}, then for any unit vector, $\linf{A x}$ can be approximated up to a multiplicative/additive error. We again consider synthetic/real-world datasets and use linear programs to obtain an upper bound on $\linf{Ax}/\linf{A_S x}$ for $x \in \text{rowspace}(A_S)$. We note that when the rows of $A$ are close to a $k$-dimensional subspace, then $A_S$ computed using Algorithm~\ref{alg:efficient} spans a subspace close to this $k$-dimensional subspace by Theorem~\ref{thm:first-version}. Hence, all the \emph{important} directions are already in $\text{rowspace}(A_S)$ and bounding $\linf{Ax}/\linf{A_S x}$ for $x \in \text{rowspace}(A_S)$ verifies that the distortion in the important directions is not large.

We construct a synthetic dataset $A = L \cdot R + G$ in a similar way to the previous section with $n=40{,}000$, $d = 10{,}000$ and $k=20$. To avoid numerical issues when solving linear programs, we now choose the coefficients of the matrices $L$ and $R$ to be i.i.d. uniform random variables drawn from $\set{-10,\ldots,10}$ and the coefficients of $G$ to be i.i.d. uniform random variables drawn from $\set{-50,\ldots,50}$. The coreset $A_S$ constructed by Algorithm~\ref{alg:efficient} for the matrix $A$ has $29$ rows and by solving $n$ linear programs, we find that $\max_{x \in \text{rowspace}(A_S)}\linf{Ax}/\linf{A_S x} \le 4.8$. 

We also perform the same experiment on the images from the previous section and find that $\linf{Ax}/\linf{A_S x} \le 1.005$ for all $x \in \text{rowspace}(A_S)$ for the first image and $
\linf{Ax}/\linf{A_S x} \le 1.03$ for all ${x \in \text{rowspace}(A_S)}$ for the second image. For real-world datasets, the coreset computed is very accurate in approximating $\linf{Ax}$ for all the interesting directions $x$. This can be explained by the fact that the value of $k$ we picked is large and the noise at that value of $k$ is small enough that many directions are \emph{covered} by the coreset and hence the coreset has a small error.
\section*{Acknowledgments}
Part of this work was done while Praneeth Kacham and David P. Woodruff were visiting Google Research. Praneeth and David were also supported in part by a Simons Investigator Award and NSF Grant No. CCF-2335412.

%% file: appendix_neurips.tex
\section{Omitted Proofs from Section~\ref{sec:linf-LRA}}
\subsection{Proof of Lemma~\ref{lma:large-implies-adds}}
Assume that there is a $k$-dimensional subspace $V$ such that $d(a_{t+1}, V)^2 > \sum_{i \in S_t}d(a_i, V)^2$ where $S_t = S \cap [t]$ is the set of rows selected by the algorithm after processing the rows $a_1, \ldots, a_t$. 

If $\sum_{i\in S_t}d(a_i, V)^2 = 0$, then $\text{rank}(A_{S_t}) \le k$ and $\text{rowspace}(A_{S_t}) \subseteq V$. Since $d(a_{t+1}, V) > 0$, we have $a_{t+1} \notin V$ which implies $a_{t+1} \notin \text{rowspace}(A_{S_t})$ and therefore the algorithm adds $t+1$ to the set $S$.

Now, suppose $\sum_{i \in S_t}d(a_i, V)^2 > 0$. Let $\Proj_V$ be the orthogonal projection matrix onto the subspace $V$ and define 
\begin{align*}
    x^* \coloneqq \frac{(I-\Proj_V)a_{t+1}}{\opnorm{(I-\Proj_V)a_{t+1}}}.
\end{align*}
Using the fact that $(I-\Proj_V)$ is also a projection matrix, we obtain
\begin{align*}
    |\la a_{t+1}, x^*\ra|^2 = \frac{(\T{a_{t+1}}(I-\Proj_V)a_{t+1})^2}{\opnorm{(I-\Proj_V)a_{t+1}}^2} = \frac{\opnorm{(I-\Proj_V)a_{t+1}}^4}{\opnorm{(I-\Proj_V)a_{t+1}}^2} = \opnorm{(I-\Proj_V)a_{t+1}}^2 = d(a_{t+1}, V)^2.
\end{align*}
We also have
\begin{align*}
    \opnorm{A_{S_t}x^*}^2 = \frac{\opnorm{A_{S_t}(I-\Proj_V)a_{t+1}}^2}{\opnorm{(I-\Proj_V)a_{t+1}}^2} \le \frac{\frnorm{A_{S_t}(I-\Proj_V)}^2\opnorm{(I-\Proj_V)a_{t+1}}^2}{\opnorm{(I-\Proj_V)a_{t+1}}^2} \le \frnorm{A_{S_t}(I-\Proj_V)}^2 = \sum_{i \in S_t}d(a_i, V)^2.
\end{align*}
Additionally, when processing the row $a_{t+1}$, the value of $\lambda$ used by the algorithm is $\frnorm{A_{S_t} - [A_{S_t}]_k}^2/k < \frnorm{A_{S_t}(I-\Proj_V)}^2/k$ since the subspace $V$ has a dimension $k$. Now, we consider two cases:
\begin{itemize}
    \item \textbf{Case 1:} $\lambda = 0$. In this case, we have $\text{rank}(A_{S_t}) \le k$. There are again two cases. If $a_{t+1} \notin \text{rowspace}(A_{S_t})$, then the algorithm adds $t+1$ to the set $S$ and we are done. 
    
    If $a_{t+1} \in \text{rowspace}(A_{S_t})$, then we can write $a_{t+1} = \T{(A_{S_t})}z$ for some $z$. If $A_{S_t}x^* = 0$, then we get $\la x^*, a_{t+1}\ra = \T{(x^*)}\T{(A_{S_t})}z = \la z, A_{S_t}x^*\ra = 0$ which contradicts our assumption that $|\la a_{t+1}, x^*\ra|^2 = d(a_{t+1}, V)^2 > \sum_{i \in S_t}d(a_i, V)^2 > 0$. Thus, $A_{S_t}x^* \ne 0$ and therefore
    \begin{align*}
        \frac{|\la a_{t+1}, x^*\ra|^2}{\opnorm{A_{S_t} x^*}^2} \ge \frac{d(a_{t+1}, V)^2}{\sum_{i \in S_t}d(a_i, V)^2} > 1.
    \end{align*}
    Finally, since $a_{t+1} \in \text{rowspace}(A_{S_t})$, we obtain $\T{a_{t+1}}(\T{A_{S_t}}A_{S_{t}})^+a_{t+1} > 1$ and therefore the algorithm adds $t+1$ to the set $S$ and we are done.
    \item \textbf{Case 2:} $\lambda \ne 0$. In this case, we have $\text{rank}(A_{S_t}) > k$ and therefore $\sum_{i \in S_t}d(a_i, V)^2 > 0$. Now,
    \begin{align*}
        \frac{{|\la a_{t+1}, x^*\ra|}^2}{\opnorm{A_{S_t}x^*}^2 + \lambda \opnorm{x^*}^2} \ge \frac{d(a_{t+1}, V)^2}{\sum_{i \in S_t}d(a_i, V)^2 + \lambda} \ge \frac{\sum_{i \in S_t}d(a_i, V)^2}{\sum_{i \in S_t}d(a_i, V)^2 + \sum_{i \in S_t}d(a_i, V)^2/k} = \frac{1}{1+1/k}.
    \end{align*}
    From the above inequality, we obtain $\T{(a_{t+1})}(\T{(A_{S_t})}A_{S_t} + \lambda I)^+ a_{t+1} > 1/(1 + 1/k)$ and therefore the algorithm adds $t+1$ to the set $S$ and we are done.
\end{itemize}
\subsection{Proof of Lemma~\ref{lma:sum-of-ridge-leverage-scores}}
Let $i^*$ be the largest index such that $\text{rank}(B_{1:i}) = k$. We note $\text{rank}(B_{1:i^*+1}) = k+1$. We now separate the sum of online rank-$k$ ridge leverage scores as
\begin{align*}
    \sum_{i=1}^n\tau_i^{\OL, k}(B) = \sum_{i=1}^{i^* + 1}\tau_i^{\OL, k}(B) + \sum_{i=i^*+2}^{n}\tau_i^{\OL, k}(B)
\end{align*}
and bound both the terms separately. Let $\text{RI}\footnote{for \textbf{R}ank \textbf{I}ncrease} \subseteq [i^*+1]$ be the set of coordinates $i$ such that $\text{rank}(B_{1:i}) > \text{rank}(B_{1:i+1})$. Note that $|\text{RI}| \le k+1$. By definition of the rank-$k$ ridge leverage scores, we have for all $i \in \text{RI}$, $\tau_i^{\OL, k}(B) = 1$. Now consider an $i < i^* + 1$ and $i \notin \text{RI}$. We have
\begin{align*}
    \tau_i^{\OL, k}(B) = \min(1, \T{b_i}(\T{(B_{1:i-1})}B_{1:i-1})^+b_i).
\end{align*}
We define $\sigma_{\min, \text{RI}} \coloneqq \min_{i \in \text{RI}}\sigma_{\min}(B_{1:i})$ where $\sigma_{\min}(\cdot)$ is used to denote the smallest \emph{nonzero} singular value of the matrix $B$. We note that for all $i \in \text{RI}$, $\opnorm{b_i} \ge \sigma_{\min, \text{RI}}$.

Now consider $i < i^*+1$ and $i \notin \text{RI}$. Note that $b_{i} \in \text{rowspace}(B_{1:i-1})$.
\begin{claim}
For $\sigma_{\min, \text{RI}}$ defined as above, the following hold:
\begin{enumerate}
    \item \begin{align*}
    \T{b_i}(\T{(B_{1:i-1})}(B_{1:i-1}))^+b_i \le 2\cdot\T{b_i}(\T{(B_{1:i-1})}B_{1:i-1} + \sigma_{\min, \text{RI}}^2 \cdot I)^+ b_i.
\end{align*}
    \item 
    \begin{align*}
        \tau_i^{\OL, k}(B) = \min(1, \T{b_i}(\T{(B_{1:i-1})}(B_{1:i-1}))^+b_i) \le 2 \cdot \min(1, \T{b_i}(\T{(B_{1:i-1})}B_{1:i-1} + \sigma_{\min, \text{RI}}^2 \cdot I)^+ b_i).
    \end{align*}
\end{enumerate}
\end{claim}
\begin{proof}
    Let $U\Sigma\T{V}$ be the ``thin'' singular value decomposition of the matrix $B_{i-1}$. It is easy to see that $\sigma_{\min}(B_{i-1}) \ge \sigma_{\min, \text{RI}}$. Since $i \notin \text{RI}$, we have $b_i \in \text{rowspace}(B_{1:i-1})$ and therefore we can write $b_i = V \cdot z$ for some $z$ which implies
    \begin{align*}
         \T{b_i}(\T{(B_{1:i-1})}(B_{1:i-1}))^+b_i = \T{z}\Sigma^{-2}\T{z}.
    \end{align*}
We can also write 
\begin{align*}
    (\T{(B_{1:i-1})}B_{1:i-1} + \sigma_{\min, \text{RI}}^2 \cdot I)^+ = V(\Sigma^2 + \sigma_{\min, \text{RI}}^2 \cdot I)^{-1}\T{V} + \frac{1}{\sigma_{\min, \text{RI}}^2}(I-V\T{V})
\end{align*}
from which we obtain
\begin{align*}
    \T{b_i}(\T{(B_{1:i-1})}B_{1:i-1} + \sigma_{\min, \text{RI}}^2 \cdot I)^+ b_i = \T{z}(\Sigma^2 + \sigma_{\min, \text{RI}}^2 \cdot I)^{-1}z \ge \frac{1}{2} \cdot \T{z}\Sigma^{-2}z = \frac{1}{2}
    \T{b_i}(\T{(B_{1:i-1})}(B_{1:i-1}))^+b_i,
\end{align*}
 where the last inequality follows from the fact that $0 \prec \Sigma^2 + \sigma_{\min, \text{RI}}^2 \cdot I \preceq 2 \cdot \Sigma^2$.

 Note that the second claim directly follows from the first.
\end{proof}
For $i \in \text{RI}$, we prove the following:
\begin{claim}
    For all $i \in \text{RI}$,
    \begin{align*}
        1 = \tau^{\OL, k}_i(B) \le \T{b_i}(\T{(B_{1:i-1})}B_{1:i-1} + \sigma_{\min, \text{RI}}^2 \cdot I)^+ b_i. 
    \end{align*}
\end{claim}
\begin{proof}
    Let $b_i^{\perp}$ be the projection of $b_i$ away from $\text{rowspace}(B_{1:i-1})$. Note that $b_i^{\perp}$ is in the rowspace of $B_{1:i}$ and therefore
    \begin{align*}
        |\la b_i, b_i^{\perp}\ra| = \opnorm{(B_{1:i}) \cdot b_i^{\perp}} \ge \sigma_{\min, \text{RI}} \cdot \opnorm{b_i^{\perp}}
    \end{align*}
    which implies
    \begin{align*}
        \frac{|\la b_i, b_i^{\perp}\ra|^2}{\opnorm{B_{1:i-1} \cdot b_i^{\perp}}^2 + \sigma_{\min, \text{RI}}^2 \opnorm{b_i^{\perp}}^2} \ge \frac{\sigma_{\min, \text{RI}}^2 \opnorm{b_i^{\perp}}^2}{0 + \sigma_{\min, \text{RI}}^2 \opnorm{b_i^{\perp}}} \ge 1.&\qedhere
    \end{align*}
\end{proof}
Thus, for all $i < i^*+1$, we have 
\begin{align*}
    \tau^{\OL, k}_i(B) \le 2 \cdot \min(1, \T{b_i}(\T{(B_{1:i-1})}B_{1:i-1} + \sigma_{\min, \text{RI}}^2 \cdot I)^+ b_i). 
\end{align*}
Hence, it suffices to bound $\sum_{i=1}^{i^*+1}\min(1, \T{b_i}(\T{(B_{1:i-1})}B_{1:i-1} + \sigma_{\min, \text{RI}}^2 \cdot I)^+ b_i)$. By Theorem~2.2 of \cite{cohen2016online}, we can bound this quantity by $O(k \log {\opnorm{B_{1:i^*+1}}}/{\sigma_{\min, \text{RI}}})$. Hence,
\begin{align*}
    \sum_{i=1}^{i^*+1}\tau_i^{\OL, k}(B) = O\left(k \log \frac{\opnorm{B_{1:i^* + 1}}}{\sigma_{\min, \text{RI}}}\right).
\end{align*}
We now want to bound
    \begin{align}
        \sum_{i=i^*+2}^{n}\tau^{\OL, k}_i(B) &= \sum_{i=i^*+2}^n \min(1, \T{b_i}(\T{B_{1:i-1}}B_{1:i-1} + \frac{\frnorm{B_{1:i-1} - [B_{1:i-1}]_k}^2}{k} \cdot I)^{-1}b_i).\label{eqn:term-to-bound-1}
    \end{align}
\citet{braverman2020near} show a bound on the $\sum_{i=1}^n \min(1, \T{b_i}(\T{B_{1:i-1}}B_{1:i-1} + \lambda I)^{-1}b_i)$ where $\lambda = \frnorm{B - [B]_k}^2/k$. The only difference in the above term we want to bound is that, instead of using a fixed $\lambda$ for all the terms as in \cite{braverman2020near}, we require an upper bound when each term has a different multiple of the identity matrix.

We will now state some useful facts, that let us use the upper bounds from \cite{braverman2020near} to bound the term in \eqref{eqn:term-to-bound-1}. Suppose $\alpha$ is such that
$\alpha/2 \le \frnorm{B_{1:i-1} - [B_{1:i-1}]_k}^2/k \le \alpha$. Then, we have from the standard properties of the L\"owner ordering that,
\begin{align*}
    \frac{1}{2}\T{B_{1:i-1}B_{1:i-1}} + \frac{\alpha}{2} \cdot I \preceq \T{B_{1:i-1}}B_{1:i-1} + \frac{\alpha}{2} \cdot I \preceq \T{B_{1:i-1}}B_{1:i-1} + \frac{\frnorm{B_{1:i-1} - [B_{1:i-1}]_k}^2}{k}\preceq \T{B_{1:i-1}}{B_{1:i-1}} + \alpha \cdot I.
\end{align*}
Since all the above matrices are positive definite, assuming $\alpha > 0$, we obtain that
\begin{align*}
    2\left(\T{B_{1:i-1}}B_{1:i-1} + \alpha \cdot I\right)^{-1} \succeq (\T{B_{1:i-1}}B_{1:i-1} + \frac{\frnorm{B_{1:i-1} - [B_{1:i-1}]_k}^2}{k} \cdot I)^{-1} \succeq (\T{B_{1:i-1}}B_{1:i-1} + \alpha \cdot I)^{-1}
\end{align*}
and therefore, 
\begin{align}
    \min(1, \T{b_i}(\T{B_{1:i-1}}B_{1:i-1} + \frac{\frnorm{B_{1:i-1} - [B_{1:i-1}]_k}^2}{k} \cdot I)^{-1}b_i) \le 2 \cdot \min(1, \T{b_i}(\T{B_{1:i-1}}B_{1:i-1} + \alpha \cdot I)^{-1}b_i).\label{eqn:int-upperbound}
\end{align}
We note that $\frnorm{B_{1:i^*+1} - [B_{1:i^*+1}]_k}^2 = \sigma_{\min}(B_{1:i^*+1})^2 \ge \sigma_{\min, \text{RI}}^2$ where we used the fact that the rank of $B_{1:i^*+1}$ is exactly $k + 1$.
For $j=1, \ldots, $ let $i_j$ be the largest $i$ such that
\begin{align*}
    \frac{\frnorm{B_{1:i-1} - [B_{1:i-1}]_k}^2}{k} \le 2^j \cdot \frac{\sigma_{\min, \text{RI}}^2}{k}
\end{align*}
and consider the intervals of integers, $(k+1 = i_0, i_1], (i_1, i_2], (i_2, i_3], \ldots$. We note that there are at most
\begin{align*}
    O\left(\log \frac{\frnorm{B - [B]_k}^2}{\sigma_{\min, \text{RI}}^2}\right) = O\left(\log \frac{\opnorm{B}}{\sigma_{\min, \text{RI}}}\right).
\end{align*}
such non-empty intervals. Now consider an arbitrary interval $(i_j, i_{j+1}]$ and we will bound
\begin{align*}
    \sum_{i \in (i_j, i_{j+1}]}\min(1, \T{b_i}(\T{B_{1:i-1}}B_{1:i-1} + \frac{\frnorm{B_{1:i-1} - [B_{1:i-1}]_k}^2}{k} \cdot I)^{-1}b_i).
\end{align*}
Setting $\alpha = 2^{j+1} \sigma_{\min, \text{RI}}^2/k$ in \eqref{eqn:int-upperbound}, we get
\begin{align*}
        &\sum_{i \in (i_j, i_{j+1}]}\min(1, \T{b_i}(\T{B_{1:i-1}}B_{1:i-1} + \frac{\frnorm{B_{1:i-1} - [B_{1:i-1}]_k}^2}{k} \cdot I)^{-1}b_i)\\
        &\le     2 \cdot \sum_{i \in (i_j, i_{j+1}]}\min(1, \T{b_i}(\T{B_{1:i-1}}B_{1:i-1} + 2^{j+1} \frac{\sigma_{\min, \text{RI}}^2}{k} \cdot I)^{-1}b_i)
\end{align*}
and since by definition $\frac{\frnorm{B_{1:i_{j+1}-1} -[B_{1:i_{j+1} - 1}]_k}^2}{k} \le 2^{j+1} \frac{\sigma^2_{\min, \text{RI}}}{k}$, we further obtain
\begin{align*}
    &\sum_{i \in (i_j, i_{j+1}]}\min(1, \T{b_i}(\T{B_{1:i-1}}B_{1:i-1} + \frac{\frnorm{B_{1:i-1} - [B_{1:i-1}]_k}^2}{k} \cdot I)^{-1}b_i)\\
    &\le \sum_{i \in (i_j, i_{j+1}]}\min(1, \T{m_i}(\T{B_{1:i-1}}B_{1:i-1} + \frac{\frnorm{B_{1:i_{j+1}-1} - [B_{1:i_{j+1}-1}]_k}^2}{k} \cdot I)^{-1}m_i).
\end{align*}
We can then finally use Lemma~2.11 of \cite{braverman2020near} to bound the above term by 
\begin{align*}
    k\log\left(1 + \frac{k\opnorm{B_{1:i_{j+1}-1}}^2}{\frnorm{B_{1:i_{j+1}-1} - [B_{1:i_{j+1}-1}]_k}^2}\right) + k + 1 \le k\log(1 + k\opnorm{B}^2/\sigma_{\min, \text{RI}}^2) + k + 1
\end{align*}
where we used the facts that $\frnorm{B_{1:i_{j+1}-1} - [B_{1:i_{j+1}-1}]_k}^2 \ge \frnorm{B_{i^*+1} - [B_{i^*+1}]_k}^2 \ge \sigma^2_{\min, \text{RI}}$ and $\opnorm{B_{1:i_{j+1}-1}}^2 \le \opnorm{B}^2$. Overall, we get that
\begin{align*}
 O(k\log(1 + k\opnorm{B}/\sigma_{\min, \text{RI}})^2) = O(k \log (k \cdot \kappa)^2).
  &\qedhere
\end{align*}
\subsection{Proof of Theorem~\ref{thm:outer-radius}}
\begin{proof}
If $V$ is a $k$-dimensional subspace and $c$ is arbitrary, then the set $V + c$ is defined as a $k$-dimensional flat. Recall that the outer $d-k$ radius of a point set $\set{a_1,\ldots,a_n} \subseteq \R^d$ is defined as
\begin{align*}
    \min_{k\text{-dim flat}\, F}\max_i d(a_i, F).
\end{align*}
Using the fact that flats are translations of $k$ dimensional subspaces, we equivalently have that the outer $d-k$ radius is equal to
\begin{align*}
    \min_{k\text{-dim subspace}\, V}\min_{c \in \R^d} \max_i d(a_i - c, V) = \min_{k\text{-dim subspace}\, V}\min_c\|(A-c)(I-\Proj_V)\|_{\infty, 2}.
\end{align*}
Here we abuse the notation and use $A-c$ to denote the matrix with rows given by $a_i - c$ for $i \in [n]$. Now define a matrix $B \doteq A - a_1$ with $n$ rows given by $0 = a_1 - a_1, a_2 - a_1, a_3 - a_2,\ldots, a_n - a_1$. For any $k$-dimensional subspace $V$ and any $c \in \R^d$, we have
\begin{align*}
    \|B(I - \Proj_V)\|_{\infty, 2} &= \|(A - a_1)(I-\Proj_V)\|_{\infty, 2}= \|(A - c + c - a_1)(I - \Proj_V)\|_{\infty, 2}\\
    &\le \|(A - c)(I- \Proj_V)\|_{\infty, 2} + \opnorm{(I-\Proj_V)(a_1 - c)}\\
    &\le 2\|(A - c)(I-\Proj_V)\|_{\infty, 2}.
\end{align*}
Hence, $\|B(I-\Proj_V)\|_{\infty, 2}\le 2\min_{c}\|(A-c)(I-\Proj_V)\|_{\infty, 2}$.
We also have $\|{B(I-\Proj_V)}\|_{\infty, 2} = \|(A - a_1)(I-\Proj_V)\|_{\infty, 2} \ge \min_c \|(A - c)(I-\Proj_V)\|_{\infty, 2}$. Thus, $\min_V \|B(I-\Proj_V)\|_{\infty, 2}$ is a $2$-approximation for $\min_{k\text{-dim flat}\ F}\max_i d(a_i, F)$ and if $S$ is the set of rows selected by Algorithm~\ref{alg:efficient} when run on the rows of the matrix $B = A - a_1$, then
\begin{align*}
    \min_V \|B_S(I-\Proj_V)\|_{\infty, 2}
\end{align*}
 is a $O(\sqrt{k}\log(n\kappa'))$ approximation for outer $(d-k)$-radius estimation of the point set $\set{a_1,\ldots,a_n}$ where $\kappa'$ is the online rank-$k$ condition number of $A - a_1$.
\end{proof}
\section{Omitted Details about Experiments}
\subsection{Measuring Distortion with in the Subspace}
Given a matrix $A$ and a parameter $k$, Algorithm~\ref{alg:efficient} returns a coreset $S$. In our experiments we measure the maximum distortion defined as $\max_{x \in \text{rowspace}(A_S)}\linf{Ax}/\linf{A_S x}$. Since any vector in the rowspace of $A_S$ can be written as $\T{A_S}y$ for some $y$, we want to measure $\max_y \linf{A\T{A_S}y}/\linf{A_S\T{A_S}y}$. Let the distortion be maximized at $y^*$ and that
\begin{align*}
    \frac{\linf{A\T{A_S}y^*}}{\linf{A_S\T{A_S}y^*}} = \phi \ge 1.
\end{align*}
Further let $i$ be the coordinate such that $\linf{A\T{A_S}y^*} = (A\T{A_S}y^*)_i$. Now for each $j \in [n]$, consider the following linear program:
\begin{align*}
    \min_{(y, t)} &\quad t\\
    \text{s.t.}\ \T{a_j}\T{A_S}y &= 1\\
    A_S\T{A_S}y &\le t \cdot 1\\
    -A_S\T{A_S}y &\le t \cdot 1.
\end{align*}
If $(y_j, t_j)$ is the optimum solution for the above problem, we note that $t_j = \linf{A_S\T{A_S}y_j}$. Since we have $\T{a_j}\T{A_S}y_j = 1$, we have that $\linf{A\T{A_S}y_j} \ge 1$ and therefore we have that $t_j = \linf{A_S\T{A_S}y_j} \ge \linf{A\T{A_S}y_j}/\phi \ge 1/\phi$. Thus for each $j \in [n]$, $1/t_j$ gives a lower bound on the maximum distortion $\phi$.

Now consider the linear program corresponding to $i \in [n]$ is defined above. Consider the vector $y = y^*/(A\T{A_S}y^*)_i$. By definition, we have $\T{a_i}\T{A_S}y = \T{a_i}\T{A_S}y^*/(A\T{A_S}y^*)_i= 1$ and $\linf{A_S\T{A_S}y} = \linf{A_S\T{A_S}y^*}/(A\T{A_S}y^*)_i = \linf{A_S\T{A_S}y^*}/\linf{A\T{A_S}y^*} = 1/\phi$. Hence, $(y, 1/\phi)$ is a feasible solution for the linear program corresponding to index $i$. Since we proved above that $t_j \ge 1/\phi$ for all $j$, we get that $t_i = 1/\phi$ and hence $\max_j 1/t_j = \phi = \max_{x \in \text{rowspace}(A_S)}\linf{A\T{A_S}y}/\linf{A_S\T{A_S}y}$. In our experiments, we solve these linear programs and find the max-distortion within the rowspace of $A_S$.
\subsection{Grayscale Images Used}
We use images from \cite{chessboard} and \cite{galaxy} for our experiments. The compressed versions of the images used are in Figure~\ref{fig:two_images}.
\begin{figure}
\centering
\begin{subfigure}[b]{0.4\textwidth}
\centering
\includegraphics[width=\textwidth]{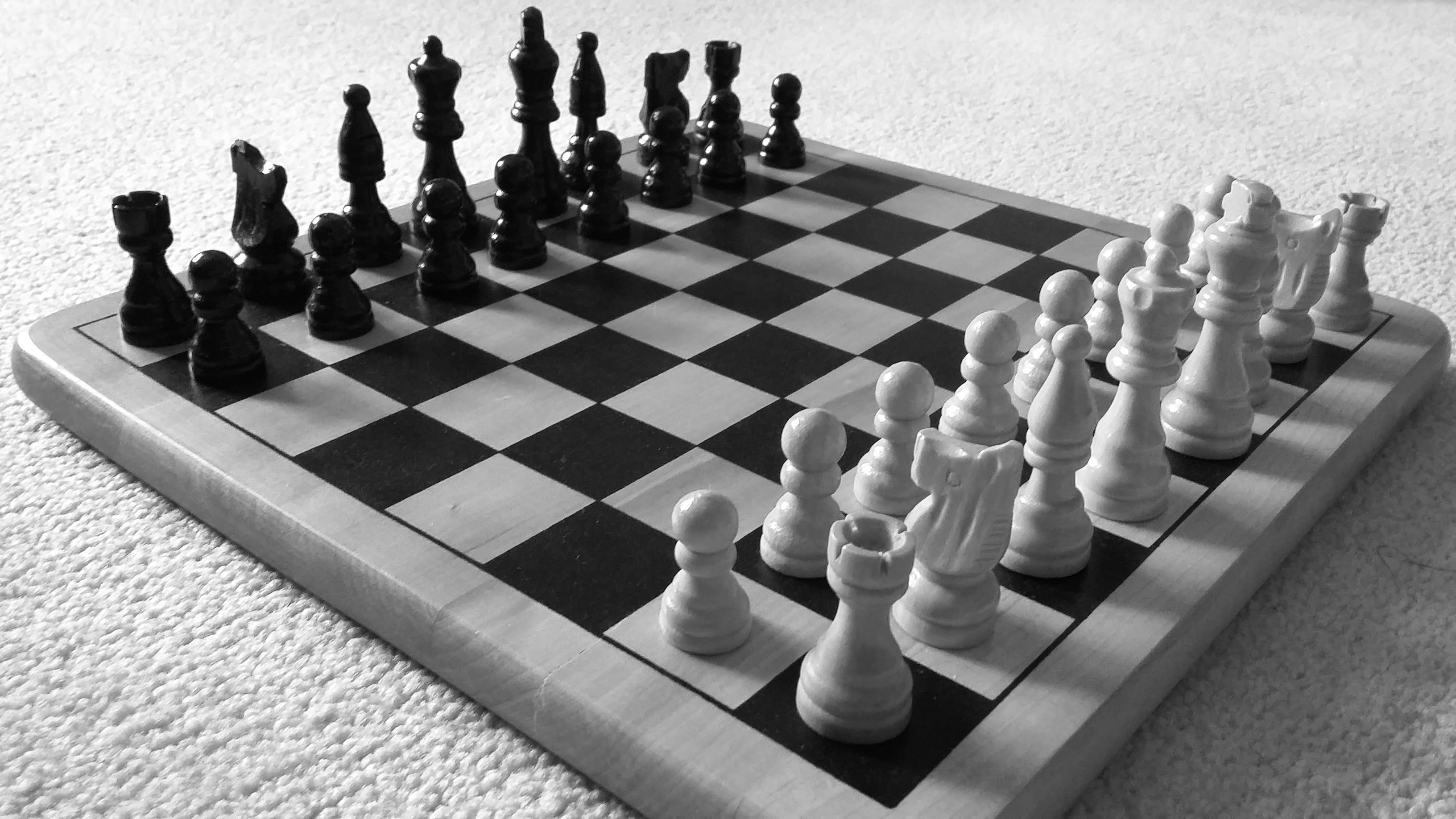}
\caption{Chessboard image from \cite{chessboard}}
\label{fig:image1}
\end{subfigure}
\hfill
\begin{subfigure}[b]{0.4\textwidth}
\centering
\includegraphics[width=\textwidth]{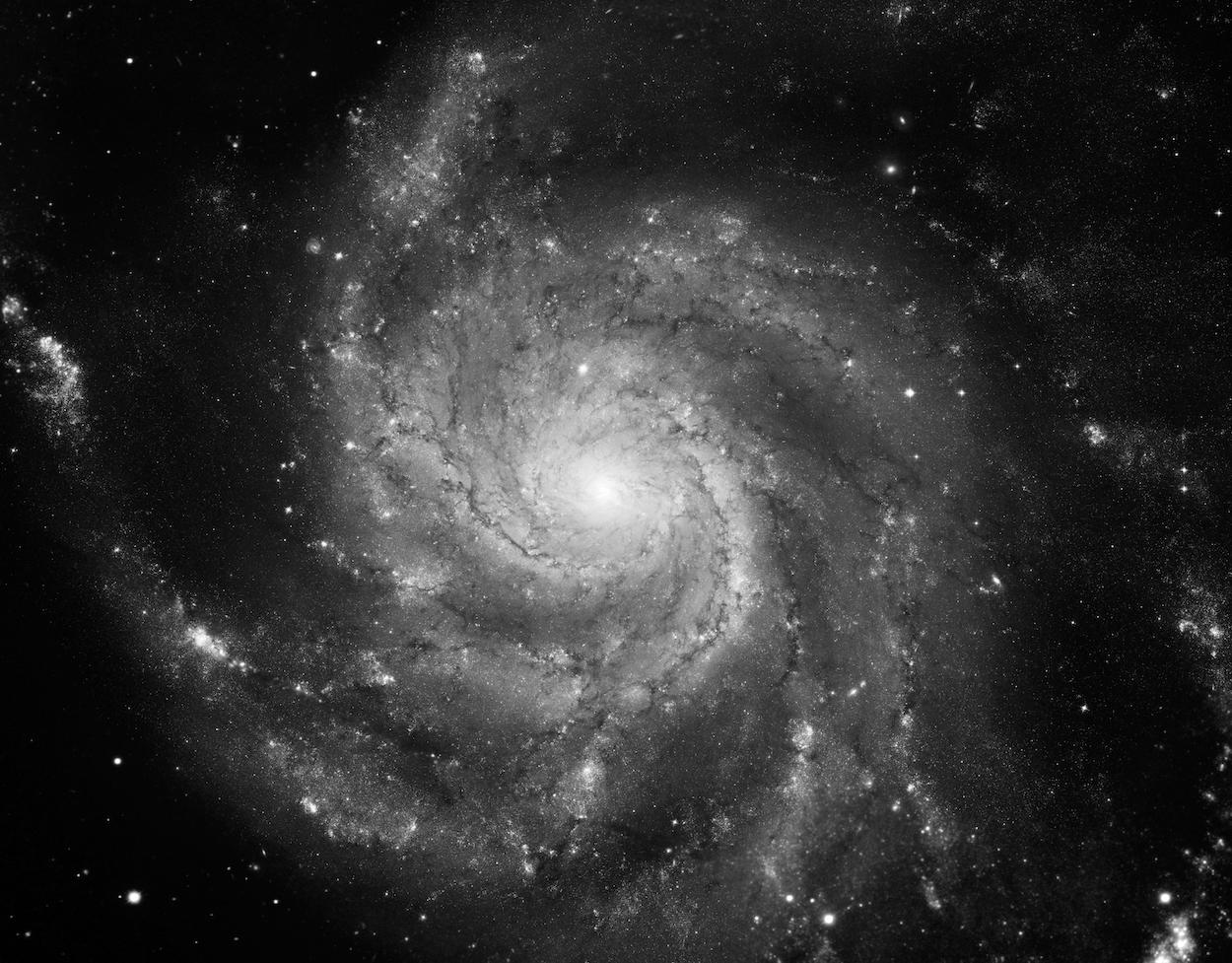}
\caption{Image of Pinwheel galaxy from \cite{galaxy}}
\label{fig:image2}
\end{subfigure}
\caption{Images used for experiments}
\label{fig:two_images}
\end{figure}